\documentclass[11pt]{article}
\usepackage{amsfonts,amsmath}
\usepackage{a4wide,amssymb}
\date{}
\newtheorem{theorem}{Theorem}
\newtheorem{acknowledgement}[theorem]{Acknowledgement}

\newtheorem{lemma}[theorem]{Lemma}

\newenvironment{proof}[1][Proof]{\noindent\textbf{#1.} }{\ \rule{0.5em}{0.5em}}
\newcommand{\R}{\mathbb{R}}

\renewcommand{\Re}{\mathrm{Re}}
\renewcommand{\Im}{\mathrm{Im}}

\let\pt=\partial

\let\e=\varepsilon
\let\f=\phi

\newcommand{\1}{\,\rlap{\small 1}\kern.13em 1}

\def \R{\mathbb{R}}
\numberwithin{equation}{section}
\numberwithin{theorem}{section}

\begin{document}

\title{\textbf{Phase Transition in a Vlasov-Boltzmann Binary Mixture}}
\author{R. Esposito$^1$, Y. Guo$^2$ and R. Marra$^3$ \\
{\footnotesize {$^1$Dipartimento di Matematica, Universit\`a di L'Aquila,
Coppito, 67100 AQ, Italy} }\\
{\footnotesize {$^2$Division of Applied Mathematics, Brown University,
Providence, RI 02812, U.S.A. } }\\
{\footnotesize {$^3$Dipartimento di Fisica and Unit\`a INFN, Universit\`a di
Roma Tor Vergata, 00133 Roma, Italy.} }}
\maketitle

\begin{abstract}
\noindent There are not many kinetic models where it is possible to prove
bifurcation phenomena for any value of the Knudsen number. Here we consider
a binary mixture over a line with collisions and long range repulsive
interaction between different species. It undergoes a segregation phase
transition at sufficiently low temperature. The spatially homogeneous
Maxwellian equilibrium corresponding to the mixed phase, minimizing the free
energy at high temperature, changes into a maximizer when the temperature
goes below a critical value, while non homogeneous minimizers, corresponding
to coexisting segregated phases, arise. We prove that they are dynamically
stable with respect to the Vlasov-Boltzmann evolution, while the homogeneous
equilibrium becomes dynamically unstable.
\end{abstract}



\section{Introduction and Main Results.}

The phenomenon of phase transition in a thermodynamic system is usually
described by the arising of multiple minimizers of the free energy. Namely,
when the temperature is lowered below a certain critical value, the unique
equilibrium minimizer becomes a local maximizer and new minimizers appear.
This is interpreted as a loss of stability of the old minimizer and the
birth of new stable states. Here the meaning of stability is merely related
to the free energy minimizing properties of the state. Of course it would be
desirable to have a more detailed, dynamical analysis of the stability
properties involved in such a phenomenon of phase transition. This issue is,
generally speaking, very difficult, because the natural dynamics is a many-component one, whose detailed understanding is very far from being
attainable. A simpler but, in our opinion, still interesting problem is to
study this kind of behavior at kinetic level, where, instead of a system of
a huge number of particles, one has to study a partial differential equation
for the probability distribution of particles on the one particle phase
space. This is the problem we want to address in this paper.

The standard kinetic model of a rarefied gas undergoing collisions
(Boltzmann equation), however, describes essentially the ideal gas, which
does not exhibit phase transitions. Following van der Waals, it is
convenient to include a long range attractive interaction between particles
to see a vapor-liquid transition. With the introduction of such an
interaction, new problems arise, due to the fact that nothing prevents the
system from collapsing (Statistical Mechanics instability). The van der
Waals approach of adding a hard core interaction is not easy to handle and
more complicated, many body long range interactions have been introduced
(see \cite{LPM,P}) to handle this in the framework of the equilibrium
Statistical Mechanics.

While one could in principle try to follow the approach in \cite{LPM}, we
find easier to consider a different kinetic model introduced a few years ago
in \cite{BL}, where there is no attractive interaction and the Statistical
Mechanics instability does not arise. The model consists of two species of
particles which, for simplicity, have the same mass. To fix the ideas, think
of them as distinguished just by their \textit{color}, say red and blue.
Their short range interaction is modeled by Boltzmann-like collisions which
are color blind, while the long range repulsive interaction, arising only between
particles of different color, is modeled by a Vlasov force with a smooth,
bounded, finite range potential. We refer to \cite{BL}, \cite{BELM},\cite{CCELM1}, \cite{CCELM2}
for more information on this model. For any $t\in \mathbb{R}^{+}$, let the
non negative normalized functions $f_{i}(t,x,\xi )$, $i=1,2$, denote the
probability densities of finding a particle of species $1$ (red) or $2$
(blue) in a cell of the phase space around the point $(x,\xi )$ at time $t$.
In this paper we will only consider the case $x\in \mathbb{R}$, the real
line, while the velocity $\xi =(v,\zeta )\in \mathbb{R}^{3}$ with $\zeta \in 
\mathbb{R}^{2}$. The time evolution is governed by 
\begin{eqnarray}
&&\pt_{t}f_{1}+v\pt_{x}f_{1}+F(f_{2})\pt_{v}f_{1}=Q(f_{1},f_{1}+f_{2}), 
\notag  \label{maineq} \\
&& \\
&&\pt_{t}f_{2}+v\pt_{x}f_{2}+F(f_{1})\pt_{v}f_{2}=Q(f_{2},f_{1}+f_{2}). 
\notag
\end{eqnarray}
Here the Vlasov force $F(h)$ due to the mass distribution $h$ is defined as 
\begin{equation}
F(h)(t,x)=-\pt_{x}\int_{\mathbb{R}}dyU(|x-y|)\int_{\mathbb{R}^{3}}d\xi
h(t,y,\xi ),  \label{vlasovforce}
\end{equation}
where $U(r)$, the interaction potential is non negative, smooth, bounded,
with $U(r)=0$ for $r\geq 1$ and 
\begin{equation*}
\int_{\mathbb{R}}dxU(|x|)=1.
\end{equation*}
The collision integral is defined as: 
\begin{eqnarray}
Q(h_{1},h_{2})(\xi ) &&=\int_{\mathbb{R}^{3}\times \mathbb{S}^{2}}|(\xi -\xi
^{\prime })\cdot \omega |\{h_{1}(\xi _{\ast })h_{2}(\xi _{\ast }^{\prime
})-h_{1}(\xi )h_{2}(\xi ^{\prime })\}d\xi ^{\prime }d\omega   \notag
\label{collisionintegral} \\
&& \\
&&:=Q_{\text{gain}}(h_{1},h_{2})-Q_{\text{loss}}(h_{1},h_{2}),  \notag
\end{eqnarray}
with $\mathbb{S}^{2}=\{\omega \in \mathbb{R}^{3}\,|\,\,|\omega |=1\}$ and $
\xi _{\ast }$, $\xi _{\ast }^{\prime }$ related to $\xi $, $\xi ^{\prime }$
by the usual elastic collision relations 
\begin{equation}
\xi _{\ast }=\xi -\omega \lbrack \omega \cdot (\xi -\xi ^{\prime })],\quad
\quad \xi _{\ast }^{\prime }=\xi ^{\prime }+\omega \lbrack \omega \cdot (\xi
-\xi ^{\prime })].
\end{equation}

For any $\beta >0$, the couple $(a_{1}\mu _{\beta },a_{2}\mu _{\beta })$,
with $\mu _{\beta }$ the spatially homogeneous Maxwellian at temperature $
\beta ^{-1}$, 
\begin{equation}
\mu _{\beta }(\xi )=\left( \frac{\beta }{2\pi }\right) ^{\frac{3}{2}}\text{e}
^{-\beta \xi ^{2}/2}  \label{stmaxwellian}
\end{equation}
and $a_{i}>0$, is an equilibrium solution and the most general homogeneous
equilibrium differs from it just for rescaling and centering. However, due
to the presence of the Vlasov force, non homogeneous Maxwellian equilibria
are possible and they are of the form 
\begin{equation}
f_{i}(x,\xi )=\rho _{i}(x)\mu _{\beta }(\xi )  \label{nstmaxwellian}
\end{equation}
provided that the densities $\rho _{i}(x)>0$ satisfy the conditions 
\begin{equation}
\log \rho _{i}+\beta U\ast \rho _{i+1}=C_{i},\quad i=1,2  \label{statcond}
\end{equation}
Here and in the rest of the paper the label $i+1$ means $2$ if $i=1$ and $1$
if $i=2$ and the convolution product $\ast $ is defined by $a\ast
b(x)=\int_{R}dya(|x-y|)b(y)$. The constants $C_{i}$ have the physical
meaning of chemical potentials. The phase transition phenomenon we want to
discuss occurs when $C_{1}=C_{2}$, a situation where the equilibrium
solutions are symmetric under the exchange $1\leftrightarrow 2$.  This symmetry
is spontaneously broken below the critical temperature. Therefore we
assume $C_{1}=C_{2}$ in the rest of the paper.

If $\beta $ is suitably small, conditions (\ref{statcond}) are satisfied
only by constant $\rho _{i}$'s. On the other hand if $\beta $ is
sufficiently large, non constant $\rho _{i}$ solving (\ref{statcond}) can
be constructed.

To understand the arising of multiple equilibria, let us start by defining
the local free energy $\varphi (\rho _{1},\rho _{2})$ on $\mathbb{R}
^{+}\times \mathbb{R}^{+}$ (i.e. the free energy density when $U$ is
replaced by a $\delta$-function) as 
\begin{equation}
\varphi (\rho _{1},\rho _{2})=\rho _{1}\log \rho _{1}+\rho _{2}\log \rho
_{2}+\beta \rho _{1}\rho _{2}  \label{locfreeenergy}
\end{equation}
Let $\rho =\rho _{1}+\rho _{2}$. It is shown in \cite{CCELM2}] (in a
slightly more general context) that, if $\beta \rho <2$ then the only
stationary points for $\varphi $ are characterized by $\rho _{1}=\rho
_{2}=\rho /2$ and they are minimizers for $\varphi $. On the other hand, if $
\beta \rho >2$, then there are $\rho ^{+}>\rho ^{-}>0$ such that 
\begin{equation}
(\rho _{1},\rho _{2})=(\rho ^{+},\rho ^{-})  \label{rho12}
\end{equation}
(and $(\rho _{1},\rho _{2})=(\rho ^{-},\rho ^{+})$ by symmetry) is an
absolute minimizer for $\varphi $, while $(\rho _{1},\rho _{2})=(\rho
/2,\rho /2)$ is a local maximizer for $\varphi $. When $(a_1,a_2)$ is chosen
equal to $(\rho ^{+},\rho ^{-})$, the equilibrium state is interpreted as a
pure phase rich of red particles, while the state $(a_{1},a_{2})=(\rho
^{-},\rho ^{+})$ corresponds to a phase rich of blue particles.

Non homogeneous solutions become relevant when one has to describe a
situation of phase coexistence. The idea to construct non homogeneous
solutions is to observe that at low temperature, in order to minimize the
local free energy, the system has to be in a pure phase at infinity. A
situation with phase coexistence should arise when the system is in the
blue-rich phase at $-\infty $ and in the red-rich phase at $+\infty $ or
vice versa. Therefore one looks at the minimizers of a suitable free energy
functional with above constraints. Indeed, the conditions (\ref{statcond})
are the Euler-Lagrange equations for the \textit{free energy} functional 
\begin{equation}
\mathcal{F}(\rho _{1},\rho _{2})=\int_{I}dx(\rho _{1}\log \rho _{1}+\rho
_{2}\log \rho _{2})+\beta \int_{I}dx\int_{I}dyU(|x-y|)\rho _{1}(x)\rho
_{2}(y).  \label{freeenergy}
\end{equation}
on a finite interval $I$ with periodic boundary conditions. Moreover, it can
be shown, (see \cite{CCELM1}), that when $\beta $ is large, with the masses
of the two species fixed, there are non constant couples $(\rho _{1},\rho
_{2})$ giving lower values to the free energy than the constants. Over the
real line $-\infty <x<\infty $, above free energy does not make sense and a
more careful definition is required: Let $(-\ell ,\ell )$ be a bounded
interval and set 
\begin{multline*}
\mathcal{F}_{(-\ell ,\ell )}(\rho _{1},\rho _{2})=\hskip-5.7pt\int_{(-\ell
,\ell )}\mathrm{d}x\varphi (\rho _{1},\rho _{2})\\+\frac{\beta }{2}
\int_{(-\ell ,\ell )^{2}}\mathrm{d}x\mathrm{d}yU(|x-y|)[\rho _{1}(x)-\rho
_{1}(y)][\rho _{2}(y)-\rho _{2}(x)].
\end{multline*}
We define the \textit{excess free energy} as 
\begin{equation}
\hat{\mathcal{F}}(\rho _{1},\rho _{2}):=\lim_{\ell \rightarrow \infty }\big[
\mathcal{F}_{(-\ell ,\ell )}(\rho _{1},\rho _{2})-2\ell \varphi (\rho
^{+},\rho ^{-})\big].  \label{excess}
\end{equation}
Note that, since $\varphi (\rho ^{+},\rho ^{-})=\varphi (\rho ^{-},\rho
^{+}) $, the functional $\hat{\mathcal{F}}$ is $+\infty $ when $\rho =(\rho
_{1},\rho _{2})$ does not go to pure phases at infinity. Since we are
interested at the coexistence of phases, we assume that $(\rho _{1},\rho
_{2})$ satisfy one of the conditions $\lim_{x\rightarrow \pm \infty }\rho
_{1}(x)=\rho ^{\pm }$ or $\lim_{x\rightarrow \pm \infty }\rho _{2}(x)=\rho
^{\mp }$. In \cite{CCELM2} several results are proved about the minimizers
for the functional $\hat{\mathcal{F}}$, which are summarized in the
following theorem (see also\cite{EGM}):

\begin{theorem}
\label{inf} Let $\beta \rho >2$. Then there exists a unique (up to
translations) positive minimizer \textrm{(front)} for the one-dimensional
excess free energy $\hat{\mathcal{F}}$, defined by (\ref{excess}), in the
class of continuous front functions $\rho =(\rho _{1}(x),\rho _{2}(x))$ such
that $\lim_{z\rightarrow \pm \infty }\rho _{1}=\rho ^{\pm },\quad
\lim_{z\rightarrow \pm \infty }\rho _{2}=\rho ^{\mp }$, where $\rho ^{\pm }$
are defined in (\ref{rho12}). We denote by $\bar{\rho}=(\bar{\rho}_{1}(x),
\bar{\rho}_{2}(x))$ the unique minimizer such that $\bar{\rho}_{1}(x)=\bar{
\rho}_{2}(-x)$. $\bar{\rho}_{1}$ is monotone increasing and $\bar{\rho}_{2}$
is monotone decreasing and $\rho ^{-}<\bar{\rho}_{i}(x)<\rho ^{+}$ for any $
x\in \mathbb{R}$. Moreover, the front $\bar{\rho}$ is $C^{\infty }(\mathbb{R}
)$-smooth and satisfies the Euler-Lagrange equations (\ref{statcond}); Its
derivative $\bar{\rho}^{\prime }$ satisfies the equations 
\begin{equation}
\bar{\rho}_{i}^{\prime }+\beta {\bar{\rho}_{i}}(U\ast \bar{\rho}
_{i+1}^{\prime })=0,\quad i=1,2.  \label{e-l'}
\end{equation}
The front $\bar{\rho}$ converges to its asymptotic values exponentially
fast, in the sense that there is $\alpha >0$ such that 
\begin{equation*}
|\bar{\rho}_{1}(x)-\rho ^{\mp }|e^{\alpha |x|}\rightarrow 0\text{ as }
x\rightarrow \mp \infty ,\quad |\bar{\rho}_{2}(x)-\rho ^{\pm }|e^{\alpha
|x|}\rightarrow 0\text{ as }x\rightarrow \mp \infty .
\end{equation*}
Finally, the derivatives of $\bar{\rho}_{i}$ of any order vanish at infinity
exponentially fast and $\bar{\rho}^{\prime }$ is odd in the sense that 
\begin{equation}
\bar{\rho}_{1}^{\prime }(x)=-\bar{\rho}_{2}^{\prime }(-x).  \label{oddness}
\end{equation}
\end{theorem}

The aim of this paper is to show that for $\beta $ larger than a certain
critical value, the non homogeneous equilibria, the front solutions $(\bar{
\rho}_{1},\bar{\rho}_{2})$ are dynamically stable with respect to the
evolution (\ref{maineq}), while the homogeneous mixed phase
is unstable. This will show that a complete bifurcation scenario arises in
this model for any value of the Knudsen number, i.e. the ratio between the
mean free path and the range of the interaction potential.

In the rest of the paper, without loss of generality we fix the asymptotic
total density $\rho =\rho ^{+}+\rho ^{-}=2$. If $\beta <1$, the only
minimizer is the homogeneous equilibrium (mixed phase) $M_{\text{hom}}(\xi
)=(\mu _{\beta }(\xi ),\mu _{\beta }(\xi ))$. On the other hand, if $\beta >1
$, the pure phases $M_{\text{red}}=(\rho ^{+}\mu _{\beta }(\xi ),\rho
^{-}\mu _{\beta }(\xi ))$ and $M_{\text{blue}}=(\rho ^{-}\mu _{\beta }(\xi
),\rho ^{+}\mu _{\beta }(\xi ))$ are constant minimizers and $M_{\text{hom}}$
is a maximizer. Moreover, we have the non homogeneous equilibrium 
\begin{equation}
M_{\bar{\rho}}=(\bar{\rho}_{1}(x)\mu _{\beta }(\xi ),\bar{\rho}_{2}(x)\mu
_{\beta }(\xi )),  \label{eqfront}
\end{equation}
with $\bar{\rho}_{1}(x)+\bar{\rho}_{2}(x)\rightarrow 2$ as $x\rightarrow \pm
\infty $.

We assume that the initial datum for the dynamics (\ref{maineq}) is a small
perturbation of one of the above equilibria, denoted generically\textbf{\ }
by $M$: 
\begin{equation}
f_{i}(0)=M_{i}+\sqrt{M_{i}}g_{i}(0),  \label{init}
\end{equation}
with $g_{i}$ sufficiently small in a sense that will be specified later.

Let $f(t,x,\xi )$ be the solution to (\ref{maineq}) and define $g(t,x,\xi )$
by setting 
\begin{equation}
f_{i}(t,x,\xi )=M_{i}(x,\xi )+\sqrt{M_{i}(x,\xi )}g_{i}(x,\xi ,t).
\label{basic-split}
\end{equation}
The equation for the perturbation $g$ is 
\begin{eqnarray}
&&\phantom{...}\left\{ \partial _{t}+v\partial _{x}+F\left( M_{i+1}+\sqrt{M_{i+1}}
g_{i+1}\right) \partial _{v}+\nu (x,\xi )\right\} g_{i}  \label{maineqpert}
\notag\\&&=\beta F\left( \sqrt{M_{i+1}}g_{i+1}\right) \sqrt{M_{i}}
v +\sum_{j=1}^{2}K^{i,j}g_{j}+F\left( \sqrt{M_{i+1}}g_{i+1}\right)
g_{i}v\\&&+\Gamma (g_{i},g_{i})+\Gamma (g_{i},g_{i+1})\notag, 
\end{eqnarray}
where we have used the notation: 
\begin{eqnarray}
&&\nu (x,\xi )=\int_{\mathbb{R}^{3}\times \mathbf{S}^{2}}|(\xi -\xi ^{\prime
})\cdot \omega |\big(M_{1}(\xi ^{\prime })+M_{2}(\xi ^{\prime })\big)d\xi
^{\prime }d\omega ,  \notag  \label{definitions} \\
&&K^{i,i}g_{i}=\frac{1}{\sqrt{M_{i}}}Q_{\text{gain }}\left( \sqrt{M_{i}}
g_{i},M_{1}+M_{2}\right) +\frac{1}{\sqrt{M_{i}}}Q\left( M_{i},\sqrt{M_{i}}
g_{i}\right) ,  \notag \\
&& \\
&&K^{i,i+1}g_{i+1}=\frac{1}{\sqrt{M_{i}}}Q\left( M_{i},\sqrt{M_{i+1}}
g_{i+1}\right) ,  \notag \\
&&\Gamma (g_{i},g_{j})=\frac{1}{\sqrt{M_{i}}}Q\left( \sqrt{M_{i}}g_{i},\sqrt{
M_{j}}g_{j}\right) .  \notag
\end{eqnarray}

We will need the following symmetry condition on the initial data: 
\begin{equation}  \label{symm}
g_1(0,x,v,\zeta)=g_2(0,-x,-v,\zeta).
\end{equation}
We note that such a property is preserved by the time evolution. It plays a
crucial role in removing the obvious orbital instability of our problem
related to the translation invariance: indeed such an invariance is
explicitly broken when condition (\ref{symm}) is assumed.

The momentum, kinetic energy and particle number of each species are
conserved during collisions, while the sum of potential and kinetic energy
is conserved along the trajectories. Therefore the following quantities are
conserved under the evolution (\ref{maineq}):

\begin{itemize}
\item The masses of the perturbation $g$, 
\begin{equation*}
{\mathcal{M}}_i(g)=\int_\R dx\int_{\mathbb{R}^3}d\xi
(f_i(x,\xi)-M_i(x,\xi)), \quad i=1,2;
\end{equation*}

\item The energy of the perturbation $g$, 
\begin{eqnarray}
\mathcal{E}(g)&=&\int_{\mathbb{R}}dx\left[\int_{\mathbb{R}^{3}}d\xi \frac{
\xi ^{2}}{ 2}\big[(f_{1}+f_{2})-(M_1+M_2)\big]\right.  \notag \\
&+&\left.\int_{\mathbb{R}}dyU(|x-y|)\big[\rho _{f_{1}}(x)\rho
_{f_{2}}(y)-\rho _{M_{1}}(x)\rho _{M_{2}}(y)\big]\right],  \label{energy}
\end{eqnarray}
with $f_i=M_i+\sqrt{M_i}g_i$ and 
\begin{equation}
\rho _{f}(x)=\int_{\mathbb{R}^{3}}d\xi f(x,\xi ),  \label{density}
\end{equation}
the spatial density of the distribution $f$.
\end{itemize}

\noindent Moreover, by standard arguments it follows that the $H$-function
of the perturbation $g$, 
\begin{eqnarray}
H(g)&=&\int_{\mathbb{R}}dx\int_{\mathbb{R}^{3}}d\xi\Big[ \big(f_{1}(x,\xi
)\log f_{1}(x,\xi )+f_{2}(x,\xi )\log f_{2}(x,\xi )\big)  \notag \\
&&-\big(M_{1}(x,\xi )\log M_{1}(x,\xi )+M_{2}(x,\xi )\log M_{2}(x,\xi )\big) 
\Big]  \label{hfunction}
\end{eqnarray}
does not increase in the evolution (\ref{maineq}).

We notice that, since the spatial domain is $\mathbb{R}$, the total masses,
energy and $H$-function of the distribution $f$ are not well defined, while
above differences are finite. Since $H(g)$ does not increase during the
evolution and $\mathcal{E}(g)$ and the total masses ${\mathcal{M}} _{i}(g)$
are constant, any linear combination of them, with positive coefficient for the entropy, does
not increase. In particular, the following non-increasing \textit{
entropy-energy} functional is crucial to study the stability of the
equilibria: 
\begin{equation}
\mathcal{H}(g)=H(g)+\beta \mathcal{E}(g)-\sum_{i=1}^{2}{\mathcal{M}}
_{i}(g)\left\{ C_i+1+\ln \left( \frac{\beta }{2\pi }\right) ^{\frac{3}{2}
}\right\} .  \label{ent-en}
\end{equation}
The factors multiplying the masses are suitably chosen to cancel some linear
terms, as will be shown in next section. The factor $\beta $ in front of the
temperature is dictated by the free energy minimizing properties of the
equilibria.

In the next sections we shall use the following weighted norm for $p\in
\lbrack 1,+\infty ]$: 
\begin{equation*}
\Vert f\Vert _{L_{w}^{p}}=\left( \int_{\mathbb{R}}dx\int_{\mathbb{R}
^{3}}d\xi |w(\xi )f(x,\xi )|^{p}\right) ^{\frac{1}{p}},
\end{equation*}
for some positive weight function $w$. Moreover $L_{w,p}$ denotes the space
of the measurable functions om $\mathbb{R}\times \mathbb{R}^{3}$ with $\Vert
\,\cdot \,\Vert _{L^{p},w}$ finite. When notation will be unambiguous, we
will also omit the index $w$. Finally $\nabla _{x,v}$ denotes the couple $
(\partial _{x},\partial _{v})$.

\begin{theorem}[Stability]
\label{stabl} Assume $\beta >1$ {and $M=M_{\bar{\rho}}$}. Let $w(\xi
)=(\Sigma +|\xi |^{2})^{\gamma }$, for $\Sigma >0$ and $\gamma >3/2$. There
are $\delta >0$ and $\Sigma >0$, $C>0$ such that, if the initial datum $
f_{i}(0)=M_{i}+\sqrt{M_{i}}g_{i}(0)$ satisfies the symmetry condition (\ref
{symm}), the bound 
\begin{equation*}
\Vert wg(0)\Vert _{L^{\infty }}+\sqrt{\mathcal{H}(g(0))}<\delta 
\end{equation*}
and $\Vert \nabla _{x,v}g(0)\Vert _{L^{2}}<+\infty $, then the initial value
problem for (\ref{maineqpert}) has a unique global in time solution with 
\begin{eqnarray}
\sup_{0\leq t\leq \infty }\Vert wg(t)\Vert _{L^{\infty }} &\leq &C\{\Vert
wg(0)\Vert _{L^{\infty }}+\sqrt{\mathcal{H}(g(0))}\},  \label{hstability} \\
\Vert \nabla _{x,v}g(t)\Vert _{L^{2}} &\leq &e^{Ct}\Vert \nabla
_{x,v}g(0)\Vert _{L^{2}}.  \label{h1bound}
\end{eqnarray}
\end{theorem}
{\bf Remark 1.1}:  {\it Theorem \ref{stabl}
implies the stability of the non homogeneous equilibrium solution $M_{\bar{
\rho}}$ in the Vlasov-Boltzmann evolution (\ref{maineq}), with respect to
initial perturbations satisfying (\ref{symm}). It turns out that $H(g(0))>0$\ for $\Vert wg(0)\Vert _{L^{\infty }}$ small (see Lemma \ref{entropy}).}

Dynamical stability and rate of convergence for the same front
solutions has been established \cite{EGM} for the Vlasov-Fokker-Planck
dynamics: the Vlasov force is the same introduced here, but the collisions
are replaced by a Fokker-Plank operator modeling the contact with a
reservoir at inverse temperature $\beta $. A careful analysis of the \textit{
macroscopic} equation plays an important role.

To prove nonlinear stability, the first major mathematical difficulty we
encounter is the presence of a large amplitude potential $U$. To our
knowledge, so far there has been no published work on stability result in
the presence of a large external field in the Boltzmann theory. The main
problem is the collapse of Sobolev estimate in higher order energy norms.
Indeed, even upon taking one $x-$derivative, the $H^{1}$-norm might actually
grow in time due to the presence of the term $\partial
_{x}F(M_{i+1})\partial _{v}g_{i}$ in (\ref{xderi}). We are thus forced to
design a strategy of proof based on a weighted $L^{\infty }$ formulation
without any derivatives and get control of derivatives only afterwords.
Furthermore, unlike the previous linear Fokker-Planck interaction, the
nonlinear Boltzmann collisions make it difficult to analyze the equation
for the projection on the hydrodynamical modes in this case. In fact, even a 
$L^{2}$ stability of $g$ is difficult to obtain directly from the analysis
of eq. (\ref{maineqpert}). Instead, we avoid a direct study of eq. (\ref
{maineqpert}) and make crucial use of the fundamental entropy-energy $
\mathcal{H}(g)$ estimate to obtain a mixed $L^{1}-L^{2}$ type of stability
estimate, based on the spectral gap of the linearized free energy operator $
A $ (Lemma \ref{entropy}). We note that the spectral gap estimate relies in
an essential way on the minimizing properties of the front solution and the
symmetry condition is used to control the part of the solution in the null
space of the operator $A$. We then bootstrap such a $L^{2}$ stability to a $
L^{\infty }$ estimate to obtain pointwise stability estimate, by following
the curved trajectory induced by the force field $F(M_{i+1}+\sqrt{M_{i+1}}
g_{i+1})$. The success of such a strategy is somewhat surprising: the
perturbation of the field does not need to decay in time, and all the
analysis over the one dimensional trajectory is carried out in a \textit{\
finite time} interval $[0,T_{0}]$ (Lemma \ref{bound}). Yet this is still
sufficient for the global in time estimate due to the strong exponential
time decay from the collision frequency.

\vskip.2cm
\noindent{\bf Remark 1.2}:  {\it The same result also holds
when perturbing initially the equilibria $M_{\text{\rm red}}$
and $M_{\text{\rm blue}}$. Indeed the proof is even simpler in
this case, because the the analog of the operator $A$ has again a spectral
gap property, but its null space is trivial. Due to this, the symmetry
assumption (\ref{symm}) is no more needed.}

We next discuss the homogeneous equilibrium $M_{\text{hom}}=(\mu _{\beta },\mu _{\beta })$.

\begin{theorem}
Assume $\beta<1$ and $M=M_{\text{\rm hom}}$. There are $\delta >0$ and $
\Sigma >0$, $C>0$ such that, if the initial datum $f_{i}(0)=M_{i}+\sqrt{M_{i}
}g_{i}(0)$ satisfies the bound 
\begin{equation*}
\Vert wg_{0}\Vert _{L^{\infty }}+\sqrt{\mathcal{H}(g(0))}<\delta
\end{equation*}
and $\Vert \nabla _{x,v}g(0)\Vert _{L^{2}}<+\infty $, then the initial value
problem for (\ref{maineqpert}) has a unique global in time solution
satisfying the estimates (\ref{hstability}) and (\ref{h1bound}).
\end{theorem}

\noindent{\bf Remark 1.3}:  {\it This implies the stability
of the mixed phase for $\beta <1$. The proof is again a simpler version of
the one for Theorem \ref{stabl} which will be omitted. In this case the
analog of operator $A$ has a spectral gap property (because $\beta <1$) and
a trivial null space, so that we do not need to require the symmetry
property (\ref{symm}).}

We have already noticed that, when $\beta >1$, the couple $(1,1)$
is a local maximizer for the free energy. In next theorem we establish the
dynamical instability of $M_{\text{hom}}$.

\begin{theorem}[Instability]
\label{ninst}Assume $\beta >1$. There exist constants $k_{0}>0$, $\theta >0$
, $C>0$, $c>0$ and a family of initial $\frac{2\pi }{k_{0}}-$periodic data $
f_{i}^{\delta }{(0)}=\mu _{\beta }+\sqrt{\mu _{\beta }}g_{i}^{\delta }{
(0)\geq 0}$, with $g^{\delta }(0)$ satisfying (\ref{symm}) and 
\begin{equation*}
\Vert \nabla _{x,v}g^{\delta }(0)\Vert _{L^{2}}+\Vert wg^{\delta }(0)\Vert
_{L^{\infty }}\leq C\delta ,
\end{equation*}
for $\delta $ sufficiently small, but the solution $g^{\delta }(t)$ to (\ref
{maineq}) satisfies 
\begin{equation*}
\sup_{0\leq t\leq T^{\delta }}\Vert wg^{\delta }(t)\Vert _{L^{\infty }}\geq
c\sup_{0\leq t\leq T^{\delta }}\Vert g^{\delta }(t)\Vert _{L^{2}}\geq
c\theta >0.
\end{equation*}
Here the escape time is
\begin{equation}
T^{\delta }=\frac{1}{\Re \lambda }\ln \frac{\theta }{\delta },
\label{tdelta}
\end{equation}
$\lambda $ is the eigenvalue with the largest real part for the linearized
Vlasov-Boltz\-mann system constructed in Theorem \ref{lininst} with $\Re 
\lambda >0$.
\end{theorem}

\noindent {\bf Remark 1.4}: {\it Note that $T^{\delta }\rightarrow
\infty $ as $\delta \rightarrow 0$. We also observe that the critical value $
\beta =1$ escapes our analysis, since it is based on some strict
inequalities which collapse for the critical $\beta $. Furthermore, the
growing mode which we construct satisfies the symmetry condition (\ref{symm}).
Hence the instability is not a consequence of the absence of such a symmetry.}

\vskip.2cm To prove instability of the homogeneous state, we encounter a
second major difficulty. Even though such a homogeneous equilibrium is not a
local minimizer for the free-energy (\ref{excess}), it is not clear at all
if such a property leads to a dynamical instability along the time
evolution. As a matter of fact, the presence of (possibly strong)
stabilizing collision in the velocity space might damp out the instability.
It is thus fundamental to understand such a stabilizing collision effect.
Unfortunately, a direct linear stability analysis in the presence of the
collision effect is too complicated to draw any conclusion. We therefore
developed a new perturbation argument to establish the linear instability
around such a homogeneous steady state. We first observe that in the absence
of the collision effect, the homogeneous state is indeed dynamically
unstable, by an explicit analysis similar to the Penrose criterion in plasma
physics \cite{Pen}. It then follows, via a perturbation argument, that for
`weak' collision effects, such an instability should persist. The key is to
show that this is even true for arbitrarily `strong' collision effect. We
use an argument of contradiction and a method of continuation. In fact, if
instability should fail at some level of the collision effect, then a
neutral mode must occur. The interaction between the Vlasov force and the
collision effect forces any neutral mode to behave like a multiple of $M_{
\text{hom}}$ with a particular dispersion relation never satisfied for our
model (Theorem \ref{lininst}). To bootstrap such a linear instability into a
nonlinear one is delicate due to severe nonlinearity, and we follow the
program developed by Strauss and the second author over the years \cite
{GS1,GS2}.

We remark that the instability of the homogeneous equilibrium for VFP model
is still open because the techniques used in the present paper cannot apply
due to the unboundedness of the Fokker-Plank operator. On the other hand, in
this paper we cannot prove the convergence to the stable equilibrium, while
it was possible in the VFP model even to compute the rate.

The paper is organized as follows: In Section 2 we use the energy-entropy to derive a mixed $L^{1}-L^{2}$
estimate. In Section 3, we establish some lemmas on the characteristics
curves for the equations (\ref{maineq}). In Section 4 we establish the
nonlinear stability in $L^{\infty }$ norm. In Section 5 we construct a
growing mode for the linearized problem around the homogeneous equilibrium
and finally, in Section 6 we show that such a linear growing mode leads to
nonlinear instability.

\section{Entropy-Energy Estimate.}

In this section we use the conservation of energy and masses, and the
entropy inequality to obtain a priori estimates on the deviation of the
solution from equilibrium. A crucial role is played by the quadratic
approximation of the excess free energy functional $\hat{\mathcal{F}}$. We
need a few definitions and some notation.

Let $u=(u_{1},u_{2})$ be a couple of functions in $L_{2}(\mathbb{R})$ and denote $
\langle u,v\rangle =\sum_{i=1}^{2}(u_i,v_i)$. Given a
couple of densities $\rho=(\rho_1(x),\rho_2(x))$, we define the operator $A$
by setting 
\begin{equation*}
\langle u,A u\rangle :=\sum_{i=1}^{2}\int_{\mathbb{R}}\mathrm{d}
xu_{i}(x)(Au)_{i}(x)=\frac{1}{2}\frac{d^{2}}{ds^{2}}\hat{\mathcal{F}}({\ \rho
}+su)\big |_{s=0}\ .
\end{equation*}
In particular, when $\rho=(\bar\rho_1(x),\bar\rho_2(x))$, the action of the
operator $A$ on $u$ is 
\begin{equation}
(Au)_{1}=\frac{u_{1}}{\bar{\rho}_{1}}+\beta U\ast u_{2},\quad (Au)_{2}=\frac{
u_{2}}{\bar{\rho}_{2}}+\beta U\ast u_{1}\ .  \label{opa}
\end{equation}
Hence 
\begin{equation*}
\langle u,Au\rangle =\frac{1}{2}\sum_{i=1}^{2}\int_{\mathbb{R}}dx\frac{
u_{i}^{2}(x)}{\bar{\rho}_{i}(x)}+\beta
\int_{R}dx\int_{R}dyU(|x-y|)u_{1}(x)u_{2}(y).
\end{equation*}
Due to the minimizing properties of $\bar{\rho}$, this quadratic form is non
negative. Moreover, by (\ref{e-l'}) and (\ref{opa}), we see that $A\bar{\rho}
^{\prime }=0$, which shows that $\bar{\rho}^{\prime }$ is in the null space
of $A$. Indeed, one can show (see \cite{EGM} and references quoted therein)
the following

\begin{lemma}
\label{spA}Suppose $\beta >1$ and $\rho =(\bar{\rho}_{1},\bar{\rho}_{2})$.
Then there exist $\delta _{0}>0$ such that 
\begin{equation*}
\langle u,Au\rangle \geq \delta _{0}\langle {(I-\mathcal{P})}u,(I-\mathcal{P}
)u\rangle ,
\end{equation*}
where ${\mathcal{P}}$ is the projector on $\mathrm{Null}\,A$: 
\begin{equation*}
\mathrm{Null}\,A=\{u\in L^{2}(\mathbb{R})\times L^{2}(\mathbb{R})\,|\,u=c
\bar{\rho}^{\prime },c\in \mathbb{R}\}.
\end{equation*}
If either $\rho =(\rho ^{\pm },\rho ^{\mp })$, or $\rho =(1,1)$ with $\beta
<1$,  we have 
\begin{equation*}
\langle u,Au\rangle \geq \delta _{0}\langle u,u\rangle ,
\end{equation*}
and the null space of $A$ reduces to $\{0\}$.\newline
\end{lemma}

We prove the following lemma which plays a crucial role in the proof of the
stability of $M_{\bar{\rho}}$. Remind that we have adopted the notation $\xi
=(v,\zeta )$, with $\zeta \in \mathbb{R}^{2}$ and $v\in \mathbb{R}$ for the
velocity.

\begin{lemma}
\label{entropy} Let $M=M_{\bar{\rho}}$. Let $f_{1}(t,x,v,\zeta
)=f_{2}(t,-x,-v,\zeta )$ and that $\Vert g\Vert _{L^{\infty }}\leq \delta $
for some small $\delta $. Then there exists $C>0$ and $\kappa >0$, such
that 
\begin{multline*}
C\sum_{i=1,2}\int_{\R}dx \int_{\R^3}d\xi\Bigg\{ \frac{(f_{i}(t)-M_{i})^{2}}{M_{i}}\mathbf{1}
_{|f_{i}(t)-M_{i}|\leq \kappa M_{i}}+\\f_{i}(t)-M_{i}|\mathbf{1}
_{|f_{i}(t)-M_{i}|\geq \kappa M_{i}}\Bigg\} \leq \mathcal{H}(g(0)).
\end{multline*}
\end{lemma}
{\bf Remark 2.1}: {\it
Note that, when dealing
with $M_{\text{\rm red}}$,  $M_{\text{\rm blue}}$ and $M_{\text{\rm hom}}$, the functions $\bar\rho_1$ and $\bar\rho_2$ have to be
replaced by the constant values $(\rho^+,\rho^-)$, $(\rho^-,\rho^+) $ and $
(1,1)$ respectively. With this modifications the lemma still holds.}

\begin{proof}
Remember the notation $\rho _{f}(t,x)=\int_{\mathbb{R}^{3}}d\xi f(t,x,\xi )$. We may construct solutions (see \cite{G4}) such that 
\begin{equation*}
\mathcal{H}(g)\leq \mathcal{H}(g(0)).
\end{equation*}
We expand $\mathcal{H}(g)$ and use (\ref{statcond}) to cancel the linear
part of the expansion, which takes the form 
\begin{eqnarray*}
&&-\sum_{i=1}^{2}\int_{\R}dx\int_{\mathbb{R}^{3}}d\xi \ \{C_{i}+1+\ln (\frac{
\beta }{2\pi })^{\frac{3}{2}}\}\{f_{i}-M_{i}\}+ \\&&\sum_{i=1}^{2}\int_{\R}dx\int_{\mathbb{R}^{3}}d\xi \Bigg\{\frac{\beta |\xi |^{2}}{2}(f_{i}-M_{i})
+\{\ln M_{i}+1\}(f_{i}-M_{i})\\&&+\beta \{U\ast \bar{\rho}
_{i+1}\}(f_{i}-M_{i})\Bigg\}.
\end{eqnarray*}
Indeed, since $\ln M_{i}=$ $\ln (\frac{\beta }{2\pi })^{\frac{3}{2}}-\frac{
\beta |\xi |^{2}}{2}+\ln \bar{\rho}_{i}$, by (\ref{statcond}), the above
quantity is zero by construction. Therefore, we turn to the second order
expansion of $\mathcal{H}(g)$. For some $\tilde{f}_{i}$ between $M_{i}$ and $
f_{i}$, 
\begin{eqnarray}
&&\mathcal{H}(g)=\sum_{i=1}^{2}\int_{R}dx\int_{\mathbb{R}^{3}}d\xi \frac{
(f_{i}(t)-M_{i})^{2}}{2\tilde{f}_{i}}dxd\xi  \\
&&+\beta \int_{\mathbb{R}}dx\int_{\mathbb{R}}dy(\rho _{f_{1}}(t,x)-\bar{\rho}
_{1}(x))U(|x-y|)(\rho _{f_{2}}(t,x)-\bar{\rho}_{2}(x)).  \notag
\label{expand2}
\end{eqnarray}
For some small number $\kappa $ to be determined, we introduce the indicator
functions $\chi _{i}^{<}=\mathbf{1}_{|f_{i}(t)-M_{i}|\leq \kappa M_{i}}$ and 
$\chi _{i}^{>}=\mathbf{1}_{|f_{i}(t)-M_{i}|>\kappa M_{i}}$ and split the
first term into 
\begin{equation*}
\int_{\mathbb{R}}dx\int_{\mathbb{R}^{3}}d\xi \frac{(f_{i}(t)-M_{i})^{2}}{2
\tilde{f}_{i}}\chi _{i}^{>}+\int_{\mathbb{R}}dx\int_{\mathbb{R}^{3}}d\xi 
\frac{(f_{i}(t)-M_{i})^{2}}{2\tilde{f}_{i}}\chi _{i}^{<}
\end{equation*}
We first estimate $\displaystyle{\frac{(f_{i}(t)-M_{i})^{2}}{2\tilde{f}_{i}}}
$ in the case of $|f_{i}(t)-M_{i}|>\kappa M_{i}$. Notice that either $
f_{i}\geq (1+\kappa )M_{i}$, or $f_{i}\leq (1-\kappa )M_{i}$. If $f_{i}\geq
(1+\kappa )M_{i}$, $\tilde{f}_{i}(t)\leq f_{i}$, and we have 
\begin{equation*}
\frac{|f_{i}(t)-M_{i}|}{\tilde{f}_{i}(t)}\geq \frac{|f_{i}(t)-M_{i}|}{f_{i}}
=1-\frac{M_{i}}{f_{i}}\geq 1-\frac{1}{1+\kappa }=\frac{\kappa }{1+\kappa }.
\end{equation*}
In the second case $f_{i}\leq (1-\kappa )M_{i}$, $\tilde{f}_{i}(t)\leq M_{i}$
and 
\begin{equation*}
\frac{|f_{i}(t)-M_{i}|}{\tilde{f}_{i}(t)}\geq \frac{|f_{i}(t)-M_{i}|}{M_{i}}
=1-\frac{f_{i}}{M_{i}}\geq 1-\{1-\kappa \}=\kappa >\frac{\kappa }{1+\kappa }.
\end{equation*}
Combining these two cases and noticing $\tilde{f}_{i}\leq (1+\kappa )M_{i}$
for $|f_{i}(t)-M_{i}|\leq \kappa M_{i}$, we conclude 
\begin{eqnarray}
&&\int_{\mathbb{R}}dx\int_{\mathbb{R}^{3}}d\xi \frac{(f_{i}(t)-M_{i})^{2}}{2
\tilde{f}_{i}}\chi _{i}^{<}+\int_{\mathbb{R}}dx\int_{\mathbb{R}^{3}}d\xi 
\frac{(f_{i}(t)-M_{i})^{2}}{2\tilde{f}_{i}}\chi _{i}^{>}  \notag \\
&\geq &\int_{\mathbb{R}}dx\int_{\mathbb{R}^{3}}d\xi \frac{
(f_{i}(t)-M_{i})^{2}}{2(1+\kappa )M_{i}}\chi _{i}^{<}+\frac{\kappa }{
2(1+\kappa )}\int_{\mathbb{R}}dx\int_{\mathbb{R}^{3}}d\xi
|f_{i}(t)-M_{i}|\chi _{i}^{>} \notag \\\label{entobound} \\
&=&\frac{1}{2(1+\kappa )}\int \int_{\mathbb{R}}dx\int_{\mathbb{R}^{3}}d\xi
g_{i}^{2}\chi _{i}^{<}+\frac{\kappa }{2(1+\kappa )}\int_{\mathbb{R}}dx\int_{
\mathbb{R}^{3}}d\xi |f_{i}(t)-M_{i}|\chi _{i}^{>}  \notag \\
&\geq &\frac{1}{2(1+\kappa )}\int \int_{\mathbb{R}}dxn_{i}^{2}+\frac{\kappa 
}{2(1+\kappa )}\int_{\mathbb{R}}dx\int_{\mathbb{R}^{3}}d\xi
|f_{i}(t)-M_{i}|\chi _{i}^{>},  \notag
\end{eqnarray}
where we have set $n_{i}(t,x)\sqrt{\mu _{\beta }}=\mathbf{P}
[(f_{i}(t)-M_{i})\chi _{i}^{<}]$, and $\mathbf{P}$ denotes the $L_{\xi }^{2}$
-projection on $\sqrt{M_{i}}$: 
\begin{equation*}
\mathbf{P}f:=\left( \int_{\mathbb{R}^{3}}d\xi ^{\prime }f(\xi ^{\prime })
\sqrt{\mu _{\beta }(\xi ^{\prime })}\right) \sqrt{\mu _{\beta }(\xi )}.
\end{equation*}

We now split the potential contribution in (\ref{expand2}) to get 
\begin{eqnarray*}
&&\ \ \ \beta \int_{\mathbb{R}}dx\int_{\mathbb{R}}dy\int_{\mathbb{R}
^{3}}d\xi \int_{\mathbb{R}^{3}}d\eta \sqrt{M_{1}}g_{1}\chi _{1}^{<}U(|x-y|)
\sqrt{M_{2}}g_{2}\chi _{2}^{<} \\
&&+\beta \int_{\mathbb{R}}dx\int_{\mathbb{R}}dy\int_{\mathbb{R}^{3}}d\xi
\int_{\mathbb{R}^{3}}d\eta \ \sqrt{M_{1}}g_{1}\chi _{1}^{<}U(|x-y|)\sqrt{
M_{2}}g_{2}\chi _{2}^{>} \\
&&+\beta \int_{\mathbb{R}}dx\int_{\mathbb{R}}dy\int_{\mathbb{R}^{3}}d\xi
\int_{\mathbb{R}^{3}}d\eta \ \sqrt{M_{1}}g_{1}\chi _{1}^{>}U(|x-y|)\sqrt{
M_{2}}g_{2}\chi _{2}^{>} \\
&&+\beta \int_{\mathbb{R}}dx\int_{\mathbb{R}}dy\int_{\mathbb{R}^{3}}d\xi
\int_{\mathbb{R}^{3}}d\eta \ \sqrt{M_{1}}g_{1}\chi _{1}^{>}U(|x-y|)\sqrt{
M_{2}}g_{2}\chi _{2}^{<}.
\end{eqnarray*}
From our assumption $\Vert g\Vert _{L^{\infty }}\leq \delta $, the last
three terms are controlled by 
\begin{eqnarray*}
&&C_{\beta }\Big(\Vert g_{1}\Vert _{L^{\infty }}+\Vert g_{2}\Vert
_{L^{\infty }}\Big)\sum_{i=2}\int_{\mathbb{R}}dx\int_{\mathbb{R}^{3}}d\xi
|f_{i}(t)-M_{1}|\chi _{i}^{>} \\
&\leq &C_{\beta }\delta \sum_{i=1}^{2}\int_{\mathbb{R}}dx\int_{\mathbb{R}
^{3}}d\xi |f_{i}(t)-M_{i}|\mathbf{1}_{|f_{i}(t)-M_{i}|\geq \kappa M_{1}}.
\end{eqnarray*}
which is bounded by the second term in (\ref{entobound}) for $\delta
<<\kappa $. Since $\int d\xi \sqrt{M_{i}}g_{i}=\sqrt{\bar{
\rho}_{i}}n_{i}$, the first term can be written as 
\begin{equation*}
\beta \int_{R}dx\int_{R}dy\ n_{1}(t,x)\sqrt{\bar{\rho}_{1}(x)}
U(|x-y|)n_{2}(t,y)\sqrt{\bar{\rho}_{2}(y)}
\end{equation*}
We now combine it with the first term in (\ref{entobound}) to get 
\begin{eqnarray*}
&&\frac{1}{2(1+\kappa )}\int_{\mathbb{R}}dx(n_{1}^{2}+n_{2}^{2})\\&&+\beta \int_{
\mathbb{R}}dx\int_{\mathbb{R}}dyn_{1}(x)\sqrt{\bar{\rho}_{1}(x)}
U(|x-y|)n_{2}(y)\sqrt{\bar{\rho}_{2}(y)}dxdy \\
&&=\langle n\sqrt{\bar{\rho}},An\sqrt{\bar{\rho}}\rangle -\frac{\kappa }{
2(1+\kappa )}\int_{R}dx(n_{1}^{2}+n_{2}^{2}).
\end{eqnarray*}
Since $n_{1}(x)=n_{2}(-x)$ by our symmetry assumption and from $\bar{\rho}
_{1}^{\prime }(x)=-\bar{\rho}_{2}^{\prime }(-x)$, it follows that $
(n_{1},n_{2})$ is orthogonal to the null space of $A$. From the spectral
inequality for the operator $A$, since $\bar{\rho}_{i}(x)\geq \rho ^{-}$,
the above is bounded from below by 
\begin{eqnarray*}
&&\left( \rho ^{-}\delta _{0}-\frac{\kappa }{2(1+\kappa )}\right)
\int_{R}dx(n_{1}^{2}+n_{2}^{2}) \\
&\geq &c\sum_{i=1}^{2}\Vert \mathbf{P}(f_{i}(t)-M_{i})\mathbf{1}
_{|f_{i}(t)-M_{i}|\leq \kappa M_{i}}\Vert _{L^{2}}^{2},
\end{eqnarray*}
provided that we choose $\kappa $ sufficiently small. Then the lemma follows
by collecting the terms with $\delta <<\kappa $.
\end{proof}

\section{1-Dimensional Characteristics.}

We define the characteristics curves $[X_{i}(s;t,x,v),V_{i}(s;t,x,v)]$ for (
\ref{maineqpert}) passing through $(t,x,v)$ at $s=t$, such that 
\begin{eqnarray}  \label{basicode}
\frac{dX_{i}(s;t,x,v)}{ds} &=&V_{i}(s;t,x,v),  \label{ode} \\
\frac{dV_{i}(s;t,x,v)}{ds} &=&-\partial _{x}U\ast \int_{\mathbb{R}^{3}}d\xi
(M_{i+1}+\sqrt{M_{i+1}}g_{i+1})\equiv -\partial _{x}\phi (X_{i}(s;t,x,v)). 
\notag
\end{eqnarray}
We also define the unperturbed characteristics $
[X_{i}^{0}(s;t,x,v),V_{i}^{0}(s;t,x,v)]$ passing through $(t,x,v)$ at $s=t$,
such that 
\begin{eqnarray}
\frac{dX_{i}^{0}(s;t,x,v)}{ds} &=&V_{i}^{0}(s;t,x,v),  \label{ode2} \\
\frac{dV_{i}^{0}(s;t,x,v)}{ds} &=&-\partial _{x}U\ast \int M_{i+1}d\xi
\equiv -\partial _{x}\phi _{0}(X_{i}^{0}(s;t,x,v))  \notag
\end{eqnarray}
Our main goal is to study the zero set of $\displaystyle{\frac{\partial
X_{i}(s;t,x,v)}{\partial v}}$.

\begin{lemma}
\label{countable}For any $(t,x,v)$ with $v\neq 0$, the set of $\{s\in
\mathbb{R}:\displaystyle{\frac{\partial X_{i}^{0}(s;t,x,v)}{\partial v}}
=0\,\}$ is countable. 
\end{lemma}

\begin{proof}
From the particle energy conservation for (\ref{ode2}): 
\begin{equation}
\frac{1}{2}|V_{i}^{0}(s;t,x,v)|^{2}-\phi _{0}(X_{i}^{0}(s;t,x,v))=\frac{1}{2}
|v|^{2}-\phi _{0}(x).  \label{particleenergy}
\end{equation}
Taking derivative with respect to $v$ yields 
\begin{equation*}
V_{i}^{0}(s;t,x,v)\frac{\partial V_{i}^{0}(s;t,x,v)}{\partial v}-\phi
_{0}^{\prime }\frac{\partial X_{i}^{0}(s;t,x,v)}{\partial v}=v.
\end{equation*}
Assume $v\neq 0$. If there is $s_{0}$ such that $\displaystyle{\frac{
\partial X_{i}^{0}(s_{0};t,x,v)}{\partial v}}=0$, then necessarily 
\begin{equation*}
V_{i}^{0}(s_{0};t,x,v)\frac{\partial V_{i}^{0}(s_{0};t,x,v)}{\partial v}\neq
0,
\end{equation*}
because, for such an $s_{0}$ it reduces to $v$ by previous equation and we
assumed $v\neq 0$. Hence 
\begin{equation*}
\displaystyle{\frac{\partial V_{i}^{0}(s_{0};t,x,v)}{\partial v}}=
\displaystyle{\frac{d}{ds}\frac{\partial X_{i}^{0}(s;t,x,v)}{\partial v}\Big|
_{s=s_{0}}}\neq 0
\end{equation*}
and therefore $\displaystyle{\frac{\partial X_{i}^{0}(s;t,x,v)}{\partial v}}
\neq 0$ in a neighborhood of $s=s_{0}$. This implies that for $v\neq 0$,
the set $\{s:\displaystyle{\frac{\partial X_{i}^{0}(s;t,x,v)}{\partial v}}
=0\}$ is countable. 
\end{proof}

\begin{lemma}
\label{determinant0}Fix $T_{0}>0$ and $N>0$. Let $|v|\leq N$.

\begin{enumerate}
\item For any $\varepsilon >0$, there exists $L_{\varepsilon }$ sufficiently
large so that, if $|x|\geq L_{\varepsilon }$, then for $0\leq s\leq
t-\varepsilon $, 
\begin{equation*}
\frac{\partial X_{i}^{0}(s;t,x,v)}{\partial v}<-\frac{\varepsilon }{2}<0.
\end{equation*}

\item For any $\eta >0$, there exist $P$ finite points $|x_{k}|\leq
L_{\varepsilon }$ ($1\leq k\leq P$) and corresponding open sets 
\begin{equation*}
O_{x_{k}}=\bigcup_{(n,o)\in I_{k}}\{a_{n}<s<b_{n}\}\times \{c_{o}<v<d_{o}\}
\end{equation*}
with the property 
\begin{equation*}
|[0,T_{0}]\times \{|v|\leq N\}\cap O_{x_{k}}^{c}|<\eta ,
\end{equation*}
so that there exists $m>0$ and, for any $|x|\leq L_{\varepsilon }$, there
exists $l\in \{1,\dots ,P\}$ 
\begin{equation*}
\Big|\frac{\partial X_{i}^{0}(s;t,x,v)}{\partial v}\Big|>m>0.
\end{equation*}
for $(s,v)\in O_{x_{l}}$.
\end{enumerate}
\end{lemma}

\begin{proof}
For any $\varepsilon >0$, from (\ref{ode2}) and (\ref{particleenergy} ) , 
\begin{eqnarray*}
|V_{i}^{0}(s;t,x,v)| &\leq &|v|+2\sqrt{\Vert \phi _{0}\Vert _{L^{\infty }}}
\leq N+C, \\
|X_{i}^{0}(s;t,x,v)-x| &\leq &T_{0}\{N+C\}.
\end{eqnarray*}
By choosing $L_{\varepsilon }$ (depending on $T_{0}$ and $N$) large enough,
for $|x|\geq L_{\varepsilon }$, $|v|\leq N$, and $0\leq s\leq T_{0}$, 
\begin{equation}
|X_{i}^{0}(s;t,x,v)|\geq \frac{L_{\varepsilon }}{2}.  \label{xlarge}
\end{equation}
From (\ref{ode2}), we have 
\begin{equation}
\frac{d^{2}}{ds^{2}}\frac{\partial X_{i}^{0}(s;t,x,v)}{\partial v}=-\partial
_{xx}\phi _{0}(X_{i}^{0}(s;t,x,v))\frac{\partial X_{i}^{0}(s;t,x,v)}{
\partial v}  \label{odexv}
\end{equation}
and we deduce that for $|s|\leq T_{0}$, 
\begin{equation}
\Big|\frac{\partial X_{i}^{0}(s;t,x,v)}{\partial v}\Big|\leq C_{T_{0}}.
\label{x0v}
\end{equation}
By the Taylor expansion for $s$, we get 
\begin{eqnarray*}
\frac{\partial X_{i}^{0}(s;t,x,v)}{\partial v} &=&\frac{\partial
X_{i}^{0}(s;t,x,v)}{\partial v}\Big|_{s=t}+(s-t)\frac{d}{ds}\frac{\partial
X_{i}^{0}(s;t,x,v)}{\partial v}\Big|_{s=t} \\
&+&\frac{(s-t)^{2}}{2}\frac{d^{2}}{ds^{2}}\frac{\partial X_{i}^{0}(\bar{s}
;t,x,v)}{\partial v} \\
&=&(s-t)+\frac{(s-t)^{2}}{2}\frac{d^{2}}{ds^{2}}\frac{\partial X_{i}^{0}(
\bar{s};t,x,v)}{\partial v}
\end{eqnarray*}
for some $t-T_{0}\leq \bar{s}\leq t$. Since the densities $\bar{\rho}_{i}$
tend to the their asymptotic values at infinity, $\lim_{y\rightarrow \infty
}\partial _{xx}\phi _{0}(y)=0$. Therefore, using again (\ref{odexv}) with $s=
\bar{s}$ and (\ref{x0v}), we have 
\begin{eqnarray*}
\frac{\partial X_{i}^{0}(s;t,x,v)}{\partial v} &\leq &(s-t)\big(
1-(t-s)C_{T_{0}}\sup_{|y|\geq \frac{L_{\varepsilon }}{2}}|\partial _{xx}\phi
_{0}(y)|\big) \\
&\leq &\frac{s-t}{2}<0,
\end{eqnarray*}
by choosing $L_{\varepsilon }$ sufficiently large. Part (1) thus follows.

To prove part (2), for $|x|\leq L_{\varepsilon }$, introduce the zero set of 
$\frac{\partial X_{i}^{0}(s;t,x,v)}{\partial v}$ as 
\begin{equation*}
Z_{x}=\{t-T_{0}\leq s\leq T_{0},|v|\leq N:\frac{\partial X_{i}^{0}(s;t,x,v)}{
\partial v}=0\}.
\end{equation*}
Then from the Fubini Theorem and Lemma \ref{countable}, 
\begin{equation*}
|Z_{x}|=\int_{-N}^{N}\left\{ \int_{t-T_{0}}^{t}\mathbf{1}_{\{s,v:\frac{
\partial X_{i}^{0}(s;t,x,v)}{\partial v}=0\}}ds\right\} dv=0.
\end{equation*}
Therefore, there exists an open set $\Omega _{x}$ such that $Z_{x}\subset
\Omega _{x}$ with $|\Omega _{x}|<\displaystyle{\frac{\eta }{2}}$. Clearly $
\frac{\partial X_{i}^{0}(s;t,x,v)}{\partial v}\neq 0$ over the compact set $
[t-T_{0},t]\times \{|v|\leq N\}\cap \Omega _{x}^{c}$. By the continuity in $
s $ and $v$, there exists $m_{x}>0$ such that over $[t-T_{0},t]\times
\{|v|\leq N\}\cap \Omega _{x}^{c}$. 
\begin{equation*}
\Big|\frac{\partial X_{i}^{0}(s;t,x,v)}{\partial v}\Big|>4m_{x}>0.
\end{equation*}
Furthermore, from the continuity in $x$, we have an open set of $(x-\Delta
_{x},x+\Delta _{x})$ such that 
\begin{equation*}
\Big|\frac{\partial X_{i}^{0}(s;t,x^{\prime },v)}{\partial v}\Big|>2m_{x}>0
\end{equation*}
for all $x^{\prime }\in (x-\Delta _{x},x+\Delta _{x})$, $(s,v)\in \lbrack
t-T_{0},t]\times \{|v|\leq N\}\cap \Omega _{x}^{c}$. Such $(x-\Delta
_{x},x+\Delta _{x})$ forms an open covering for $|x|\leq L_{\varepsilon }$,
hence there is a finite subcovering $\{(x_{k}-\Delta _{k},x_{k}+\Delta
_{k}),k=1,\dots ,P\}$ for $|x|\leq L_{\varepsilon }$. For any $|x|\leq
L_{\varepsilon }$, there exists $\Omega _{x_{l}}$ such that $x\in
(x_{l}-\Delta _{l},x_{l}+\Delta _{l})$ and for $[t-T_{0},t]\times \{|v|\leq
N\}\cap \Omega _{x_{l}}^{c}$, 
\begin{equation*}
\Big|\frac{\partial X_{i}^{0}(s;t,x,v)}{\partial v}\Big|>2m=\min_{1\leq
k\leq P}2m_{k}>0.
\end{equation*}
We finally choose an (finite) open covering of $[t-T_{0},t]\times \{|v|\leq
N\}\cap \Omega _{x_{k}}^{c}$ of the form $O_{x_{k}}=\bigcup_{n,o}
\{a_{n}<s<b_{n}\}\times \{c_{o}<v<d_{o}\}$ with \thinspace $
|a_{n}-b_{n}|+|c_{o}-d_{o}|$ sufficiently small so that over $O_{x_{k}}$ 
\begin{equation*}
\left\vert \frac{\partial X_{i}^{0}(s;t,x,v)}{\partial v}\right\vert >m>0
\end{equation*}
\end{proof}

\begin{lemma}
\label{determinant}Fix $T_{0}>0$ and $N>0$. Let $|v|\leq N$. For any $
\varepsilon >0$, recall $L_{\varepsilon }$, $O_{x_{k}}$, $1\leq k\leq P$
constructed in Lemma \ref{determinant0} . There exists $\delta >0$ such that
if $\Vert g\Vert _{L^{2}}<\delta $,

\begin{enumerate}
\item If $|x|\geq L_{\varepsilon }$ then for $0\leq s\leq t-\varepsilon $, 
\begin{equation*}
\frac{\partial X_{i}(s;t,x,v)}{\partial v}<-\frac{\varepsilon }{2}.
\end{equation*}

\item For any $|x|\leq L_{\varepsilon }$, there exists $l$ such that $x\in
(x_{l}-\Delta _{l},x_{l}+\Delta _{l})$ and for all $(s,v)\in O_{x_{l}}$ 
\begin{equation*}
\left\vert \frac{\partial X_{i}(s;t,x,v)}{\partial v}\right\vert >\frac{m}{2}
>0.
\end{equation*}
\end{enumerate}
\end{lemma}

\begin{proof}
Denote the solution operator of (\ref{ode2}) and (\ref{ode}) by $G_{0}$ and $
G$. By the Duhamel principle, we have 
\begin{multline}
\binom{X_{i}(s;t,x,v)}{V_{i}(s;t,x,v)}=e^{G_{0}(s-t)}\binom{x}{v}
\\+\int_{t}^{s}e^{G_{0}(s-\tau )}\binom{0}{-\partial _{x}U\ast \{\int \sqrt{
M_{i+1}}g_{i+1}d\xi ^{\prime }\}(\tau )}d\tau .  \label{xv}
\end{multline}
Notice that 
\begin{eqnarray*}
&&\phantom{...}\left\vert \frac{\partial }{\partial v}\partial _{x}U\ast \{\int \sqrt{
M_{i+1}}g_{i+1}d\xi ^{\prime }\}(\tau )\right\vert \\&&=\left\vert \partial
_{xx}U\ast \{\int \sqrt{M_{i+1}}g_{i+1}d\xi ^{\prime }\}\frac{\partial
X_{i}(\tau ;t,x,v)}{\partial v}\right\vert \\
&&\le C\Vert g\Vert _{L^{2}}\left\vert \frac{\partial X_{i}(\tau ;t,x,v)}{
\partial v}\right\vert .
\end{eqnarray*}
It thus follows from taking derivative of $v$ in (\ref{xv}) and by the
Gronwall's lemma that $0\leq s\leq T_{0}$ 
\begin{equation*}
\left\vert \frac{\partial X_{i}(s;t,x,v)}{\partial v}\right\vert \leq
e^{CT_{0}}.
\end{equation*}
We now use (\ref{xv}) again to get 
\begin{equation*}
\left\vert \frac{\partial X_{i}(s;t,x,v)}{\partial v}-\frac{\partial
X_{i}^{0}(s;t,x,v)}{\partial v}\right\vert \leq C_{T_{0}}\Vert g\Vert
_{L^{2}}.
\end{equation*}
Hence, we deduce our lemma by choosing $\Vert g\Vert _{L^{2}}$ sufficiently
small.
\end{proof}

\section{Weighted $L^{\infty }$ Stability.}

In this section we use the entropy-energy bound and the estimates on the
characteristics to show that the perturbation $g$ of the non homogeneous
equilibrium $M_{\bar{\rho}}$, is arbitrarily small at any positive time in a
suitable weighted $L_{\infty }$ norm, provided that it is initially
sufficiently small, thus showing the stability of the non homogeneous
equilibrium. We use the weight function $w(\xi )=\big(\Sigma +|\xi |^{2}\big)
^{\gamma }$. with $\Sigma $ a positive constant to be chosen later and $
\gamma >\frac{3}{2}$.

\begin{lemma}
\label{bound}Let $h=wg$. There exist $T_{0}>0$ and $\delta >0$ such that, if 
$\Vert h\Vert _{L^{\infty }}<\delta $, then 
\begin{equation*}
\Vert h(T_{0})\Vert _{L^{\infty }}\leq \frac{1}{2}\Vert h(0)\Vert
_{L^{\infty }}+C_{T_{0}}\sqrt{\mathcal{H}(g(0))}.
\end{equation*}
\end{lemma}

\begin{proof}
We first write the equation for $h=wg$ from (\ref{maineqpert}): 
\begin{eqnarray}  \label{linear}
&&\left( \partial _{t}+v\partial _{x}+F(M_{i+1}+\sqrt{M_{i+1}}\frac{h_{i+1}}{
w}\right) \partial _{v}+\nu (x,\xi )\Big)h_{i}=  \notag  \label{h} \\
&&F\left( M_{i+1}+\sqrt{M_{i+1}}\frac{h_{i+1}}{w}\right) \frac{w^{\prime }}{
w}h_{i}+wF\left( \sqrt{M_{I+1}}g_{i+1}\right) v\sqrt{M_{i}} +\sum_{j=1,2}K_{w}^{ij}h_{j}\notag\\
&&+F\left( \sqrt{M_{i+1}}\frac{h_{i+1}}{w}
\right) vh_{i}+w\Gamma \left( \frac{h_{i}}{w},\frac{h_{i}}{w}\right)
+w\Gamma \left( \frac{h_{i}}{w},\frac{h_{i+1}}{w}\right) ,
\end{eqnarray}
where ${K}_{w}^{ij}(\cdot )=wK^{ij}\left( \displaystyle{\frac{\cdot }{w}}
\right) $. We note that for $\Sigma \geq 1$ 
\begin{equation}
\frac{w(\xi )}{w(\xi ^{\prime })}=\frac{[\Sigma +|\xi |^{2}]^{\gamma }}{
[\Sigma +|\xi ^{\prime }|^{2}]^{\gamma }}\leq C_{\gamma }\frac{[\Sigma +|\xi
^{\prime }|^{2}]^{\gamma }+|\xi ^{\prime }-\xi |^{2\gamma }}{[\Sigma +|\xi
^{\prime }|^{2}]^{\gamma }}\leq C_{\gamma }[1+|\xi ^{\prime }-\xi
|^{2}]^{\gamma }.  \label{weight}
\end{equation}
For any $(t,x,\xi )$, integrating along its backward trajectory (\ref{ode})
\begin{equation*}
\lbrack X_{i}(s),V_{i}(s)]=[X_{i}(s;t,x,v),V_{i}(s;t,x,v)]
\end{equation*}
we can express $h_{i}(t,x,\xi )$ as 
\begin{eqnarray}
&&h_{i}(0,X_{i}(0;t,x,v),V_{i}(0;t,x,v),\zeta )+  \notag \\
&&\int_{0}^{t}e^{\int_{t}^{s}\nu _{i}(\tau )d\tau }\{F(M_{i+1}+\sqrt{M_{i+1}}
g_{i+1})\frac{w^{\prime }}{w}h_{i}\}(s,X_{i}(s),V_{i}(s),\zeta )ds  \notag \\
&&+\int_{0}^{t}e^{\int_{t}^{s}\nu _{i}(\tau )d\tau }\{F(\sqrt{M_{i+1}}
g_{i+1})V_{i}(s)\sqrt{M_{i}}\}(s,X_{i}(s),V_{i}(s),\zeta )ds  \notag \\
&&+\sum_{j=1}^{2}\int_{0}^{t}e^{\int_{t}^{s}\nu _{i}(\tau )d\tau }\left(
\sum_{j=1,2}K_{w}^{i,j}h_{j}\right) (s,X_{i}(s),V_{i}(s),\zeta )ds  \notag \\
&&+\int_{0}^{t}e^{\int_{t}^{s}\nu _{i}(\tau )d\tau }\{F(\sqrt{M_{i+1}}
g_{i+1})vh_{i}\}(s,X_{i}(s),V_{i}(s),\zeta )  \notag \\
&&+\int_{0}^{t}e^{\int_{t}^{s}\nu _{i}(\tau )d\tau }w\Big[\Gamma \Big(\frac{
h_{i}}{w},\frac{h_{i}}{w}\Big)+\Gamma \Big(\frac{h_{i}}{w},\frac{h_{i+1}}{w}\Big)\Big](s,X_{i}(s),V_{i}(s),\zeta )ds.  \label{duhamel}
\end{eqnarray}

We have set $\nu _{i}(\tau )\equiv \nu (V_{i}(\tau ),\zeta )\geq \nu _{0}>0$
. Fix a small constant $\varepsilon >0$. We can choose $\Sigma $ large so
that $|\frac{w^{\prime }}{w}|\leq \varepsilon $. Since $\Vert F(M_{i+1}+
\sqrt{M_{i+1}}g_{i+1})\Vert _{L^{\infty }}\leq C$ if $\Vert h\Vert
_{L^{\infty }}$ is small, the second term in (\ref{duhamel}) is bounded by 
\begin{equation*}
C\varepsilon e^{-\frac{\nu _{0}t}{2}}\sup_{0\leq s\leq T_{0}}\{e^{\frac{\nu
_{0}s}{2}}\Vert h(s)\Vert _{L^{\infty }}\}.
\end{equation*}

For the third term in (\ref{duhamel}), we split 
\begin{equation*}
g_{i+1}=g_{i+1}\mathbf{1}_{|f_{i}(t)-M_{i}|\geq \kappa M_{i}}+g_{i+1}\mathbf{
1}_{|f_{i}(t)-M_{i}|\geq \kappa M_{i}}.
\end{equation*}
Since $U$ is smooth, 
\begin{multline*}
\Vert F(\sqrt{M_{i+1}}g_{i+1})\Vert _{L^{\infty }}\\\leq C\{\Vert \sqrt{M_{i+1}
}g_{i+1}\mathbf{1}_{|f_{i}(t)-M_{i}|\leq \kappa M_{i}}\Vert _{L^{2}}+\Vert 
\sqrt{M_{i+1}}g_{i+1}\mathbf{1}_{|f_{i}(t)-M_{i}|\geq \kappa M_{i}}\Vert
_{L^{1}}\},
\end{multline*}
by Lemma \ref{entropy}, 
\begin{eqnarray}
&&\Vert F(\sqrt{M_{i+1}}g_{i+1})V_{i}(s)\sqrt{M_{i}}(s,X_{i}(s),V_{i}(s),
\zeta )\Vert _{L^{\infty }}  \notag \\
&\leq &C\big(\mathcal{H}(g(0))+\sqrt{\mathcal{H}(g(0))}\big).  \label{split}
\end{eqnarray}

For the fifth term, we note that, since for hard spheres $\nu _{i}(s)\geq
\nu _{0}(1+|\zeta |+|V_{i}(s)(s)|$, it follows that $\displaystyle{\frac{
|V_{i}(s)|}{\nu _{i}(s)}<\frac{1}{\nu _{0}}}$. Moreover, $\displaystyle{
\int_{0}^{t}ds\frac{\nu _{i}(s)}{2}\text{\textrm{e}}^{\int_{t}^{s}(\nu
_{i}(\tau )/2)d\tau }\leq 1}$. Therefore, 
\begin{equation*}
\left\vert \int_{0}^{t}e^{\int_{t}^{s}\nu _{i}(\tau )d\tau }F(\sqrt{M_{i+1}}
g_{i+1})V_{i}(s)h_{i}ds\right\vert \leq Ce^{-\nu _{0}t}\sup_{0\leq s\leq
T_{0}}\{e^{\frac{\nu _{0}s}{2}}\Vert h(s)\Vert _{L^{\infty }}\}^{2}.
\end{equation*}
For the last term in (\ref{duhamel}), by Lemma 10 of \cite{G2}, it follows 
\begin{equation*}
\left\vert w\Gamma \left( \frac{h_{i}}{w},\frac{h_{i}}{w}\right) (\xi
)\right\vert +\left\vert w\Gamma \left( \frac{h_{i}}{w},\frac{h_{i+1}}{w}
\right) (\xi )\right\vert \leq C\nu (\xi )\Vert h\Vert _{L^{\infty }}^{2}.
\end{equation*}
We therefore get the bound for the last term by 
\begin{multline*}
\int_{0}^{t}e^{\int_{t}^{s}\nu (\tau )d\tau }\nu _{i}(s)\Vert h(s)\Vert
_{L^{\infty }}^{2}ds\\ \leq C\{\sup_{0\leq s\leq T_{0}}e^{\frac{\nu _{0}s}{2}
}\Vert h(s)\Vert _{L^{\infty }}\}^{2} \int_{0}^{t}e^{\int_{t}^{s}\nu
_{i}(\tau )d\tau }\nu _{i}(s)e^{-\nu _{0}s}ds.
\end{multline*}
Note that $\frac{d}{ds}[e^{\int_{t}^{s}\nu _{i}(\tau )d\tau
}]=e^{\int_{t}^{s}\nu _{i}(\tau )d\tau }\nu _{i}(s)$. Integrating by parts
yields 
\begin{eqnarray*}
\int_{0}^{t}e^{\int_{t}^{s}\nu _{i}(\tau )d\tau }\nu _{i}(s)e^{-\nu _{0}s}ds
&=&\Big(e^{\int_{t}^{s}\nu _{i}(\tau )d\tau }e^{-\nu _{0}s}\Big)\Big|
_{s=0}^{s=t}+\nu _{0}\int_{0}^{t}e^{\int_{t}^{s}\nu _{i}(\tau )d\tau
}e^{-\nu _{0}s}ds \\
&\leq &C(1+t)e^{-\nu _{0}t}.
\end{eqnarray*}

We shall mainly concentrate on the fourth term in (\ref{duhamel}). Let $
\mathbf{k}^{i,j}(\xi ,\xi ^{\prime })$ be the corresponding kernel
associated with $K_{w}^{ij}$ in (\ref{linear}). We now use (\ref{duhamel})
for $h_{j}(s,X_{i}(s),\xi ^{\prime })$ again to evaluate 
\begin{equation*}
\{K_{w}^{i,j}h_{j}\}(s,X_{i}(s),V_{i}(s),\zeta )=\int \mathbf{k}
_{w}^{i,j}(V_{i}(s),\zeta ,\xi ^{\prime })h_{j}(s,X_{i}(s),\xi ^{\prime
})d\xi ^{\prime }.
\end{equation*}
Denote 
\begin{equation*}
\lbrack X_{j}(s_{1}),V_{j}(s_{1})]\equiv \lbrack
X_{j}(s_{1};X_{i}(s;t,x,v),v^{\prime }),V_{j}(s_{1};X_{i}(s;t,x,v),v^{\prime
})].
\end{equation*}
We can bound the fourth term in (\ref{duhamel}) by the sum on $j$ of
\begin{eqnarray}
&&\int_{0}^{t}e^{\int_{t}^{s}\nu _{i}(\tau )d\tau +\int_{s}^{0}\nu _{j}(\tau
)d\tau }\int_{\ \mathbb{R}^{3}}|\mathbf{k}_{w}^{i,j}(V_{i}(s),\zeta ,\xi
^{\prime })|h_{j}(0,X_{j}(0),V_{j}(0),\zeta ^{\prime })|d\xi ^{\prime }ds 
\notag \\
&&+\int_{0}^{t}\int_{s_{1}}^{s}e^{\int_{t}^{s}\nu _{i}(\tau )d\tau
+\int_{s}^{s_{1}}\nu _{j}(\tau )d\tau }\int_{\ \mathbb{R}^{3}}|\mathbf{k}
_{w}^{i,j}(V_{i}(s),\zeta ,\xi ^{\prime })|  \notag \\
&&\times \{F(M_{j+1}+\sqrt{M_{j+1}}g_{j+1})\frac{w^{\prime }}{w}
h_{j}\}(s_{1},X_{j}(s_{1}),V_{j}(s_{1}),\zeta ^{\prime })d\xi ^{\prime
}dsds_{1}  \notag \\
&&+\int_{0}^{t}\int_{s_{1}}^{s}e^{\int_{t}^{s}\nu _{i}(\tau )d\tau
+\int_{s}^{s_{1}}\nu _{j}(\tau )d\tau }\int_{\ \mathbb{R}^{3}}|\mathbf{k}
_{w}^{i,j}(V_{i}(s),\zeta ,\xi ^{\prime })|  \notag \\
&&\times \{F(\sqrt{M_{j+1}}g_{j+1})v\sqrt{M_{j}}
\}(s_{1},X_{j}(s_{1}),V_{j}(s_{1}),\zeta ^{\prime })d\xi ^{\prime }ds  \notag
\\
&&+\sum_{k}\int_{0}^{t}\int_{s_{1}}^{s}\int_{\ \mathbb{R}^{3}\times \ 
\mathbb{R}^{3}}e^{\int_{t}^{s}\nu _{i}(\tau )d\tau +\int_{s}^{s_{1}}\nu
_{j}(\tau )d\tau }\big|\mathbf{k}_{w}^{i,j}(V_{i}(s),\zeta ,\xi ^{\prime })\notag\\&&\times
\mathbf{k}_{w}^{j,k}(V_{j}(s_{1}),\zeta ^{\prime },\xi ^{\prime \prime })\big| 
  h_{k}(s_{1},X_{j}(s_{1}),\xi ^{\prime \prime })|d\xi ^{\prime }d\xi
^{\prime \prime }dsds_{1}  \notag \\
&&+\int_{0}^{t}\int_{s_{1}}^{s}e^{\int_{t}^{s}\nu _{i}(\tau )d\tau
+\int_{s}^{s_{1}}\nu _{j}(\tau )d\tau }\int_{\ \mathbb{R}^{3}}|\mathbf{k}
_{w}^{i,j}(V_{i}(s),\zeta ,\xi ^{\prime })|  \notag \\
&&\times \{F(\sqrt{M_{j+1}}g_{j+1})vh_{j}\}(s_{1},X_{j}(s_{1}),V_{j}(s_{1}),
\zeta ^{\prime })d\xi ^{\prime }dsds_{1}  \notag \\
&&+\int_{0}^{t}\int_{s_{1}}^{s}e^{\int_{t}^{s}\nu _{i}(\tau )d\tau
+\int_{s}^{s_{1}}\nu _{j}(\tau )d\tau }\int_{\ \mathbb{R}^{3}}|\mathbf{k}
_{w}^{i,j}(V_{i}(s),\zeta ,\xi ^{\prime })|  \label{double} \\
&&\times w\Big\{\Gamma \Big(\frac{h_{j}}{w},\frac{h_{j}}{w}\Big)+\Gamma \Big(
\frac{h_{j}}{w},\frac{h_{j+1}}{w}\Big)\Big\}(s,X_{j}(s_{1}),V_{j}(s_{1}))d
\xi ^{\prime }dsds_{1}.  \notag
\end{eqnarray}

We will make an extended use of Lemma 7 of \cite{G2}, which we report here
for reader's convenience: For hard spheres, the usual Grad estimates imply: 
\begin{equation}
|\mathbf{k}_{i,j}(\xi,\xi^{\prime})|\le
C\{|\xi-\xi^{\prime}|+|\xi-\xi^{\prime}|^{-1}\}\text{\textrm{e}}^{-\frac 1 8
|\xi-\xi^{\prime}|^2-\frac 1 8 \frac{\|\xi|^2-|\xi^{\prime}|^2|^2}{
|\xi-\xi^{\prime}|^2}}.
\end{equation}

\begin{lemma}[Lemma 7 of \protect\cite{G2}]
\label{lemma7}There are $\e>0$ and $C>0$ such that 
\begin{equation}
\int_{\mathbb{R}^{3}}d\xi ^{\prime }\frac{w(\xi )}{w(\xi ^{\prime })}\{|\xi
-\xi ^{\prime }|+|\xi -\xi ^{\prime }|^{-1}\}\text{\textrm{e}}^{-\frac{1-\e}{
8}|\xi -\xi ^{\prime }|^{2}-\frac{1-\e}{8}\frac{\Vert \xi |^{2}-|\xi
^{\prime }|^{2}|^{2}}{|\xi -\xi ^{\prime }|^{2}}}\leq \frac{C}{1+|\xi |}.
\end{equation}
\end{lemma}

By Lemma \ref{lemma7}, we obtain the crucial estimate 
\begin{equation}
\int_{\ \mathbb{R}^{3}}\mathbf{k}_{w}^{i.j}|(\xi ,\xi ^{\prime })|d\xi
^{\prime }<\frac{C}{1+|\xi |}  \label{integral}
\end{equation}
uniformly in $\Sigma $. 
Since $\nu _{i}\geq \nu _{0}$, by taking $L^{\infty }$ norm for $h$ and (\ref
{integral}), we bound the first term in (\ref{double}) by $Cte^{-\nu
_{0}t}\Vert h_{0}\Vert _{L^{\infty }}$, and the second term by 
\begin{equation*}
\varepsilon Ce^{-\frac{\nu _{0}}{2}t}\sup_{0\leq s\leq T_{0}}\{e^{\frac{\nu
_{0}}{2}s}\Vert h(s)\Vert _{L^{\infty }}\}.
\end{equation*}
By (\ref{integral}), the third term is bounded by $C\mathcal{H}(g(0))$ as in
(\ref{split}), and the last two nonlinear terms are bounded by 
\begin{equation*}
C\{1+t\}e^{-\nu _{0}t}\{\sup_{0\leq s\leq T_{0}}e^{\frac{\nu _{0}}{2}s}\Vert
h(s)\Vert _{L^{\infty }}^{{}}\}^{2}.
\end{equation*}

We now concentrate on the fourth term in (\ref{double}), which will be
estimated along the same lines of the proof of Theorem 20 in \cite{G2}.

\textbf{CASE 1:} For $|\xi |\geq N_{T_{0}}$, we know that from (\ref{ode}) 
\begin{equation*}
|V_{i}(s;t,x,v)-v|\leq |s-t|C\leq CT_{0},
\end{equation*}
By Lemma \ref{lemma7} and (\ref{weight}), for $N_{T_{0}}$ large 
\begin{multline*}
\int \int |\mathbf{k}_{w}^{i,j}(V_{i}(s),\zeta ,\xi ^{\prime })\mathbf{k}
_{w}^{j,k}(V_{j}(s_{1}),\zeta ^{\prime },\xi ^{\prime \prime })|d\xi
^{\prime }d\xi ^{\prime \prime }\\\leq \frac{C_{{}}}{1+|V_{i}(s)|+|\zeta |}
\leq \frac{C}{N-CT_{0}},
\end{multline*}
we therefore can find an upper bound for the fourth term in (\ref{double})
by ($N>>T_{0}$) 
\begin{multline*}
\frac{C}{N}\int_{0}^{t}e^{-\nu _{0}(t-s)}\times \int_{0}^{s}e^{-\nu
_{0}(s-s_{1})}\Vert h(s_{1})\Vert _{L^{\infty }}ds_{1}ds\\\leq \frac{Ce^{-
\frac{\nu _{0}}{2}t}}{N}\sup_{0\leq s\leq T_{0}}e^{\frac{\nu _{0}}{2}s}\Vert
h(s)\Vert _{L^{\infty }}.
\end{multline*}

\textbf{CASE 2:}\textit{\ }For $|\xi |\leq N$, $|\xi ^{\prime }|\geq 2N$, or 
$|\xi ^{\prime }|\leq 2N$, $|\xi ^{\prime \prime }|\geq 3N$. Notice that we
have either $|\xi ^{\prime }-\xi |\geq N$ or $|\xi ^{\prime }-\xi ^{\prime
\prime }|\geq N$. This implies that 
\begin{eqnarray*}
|v^{\prime }-V_{i}(s;t,x,v)| &\geq &|v^{\prime }-v|-|v-V_{i}(s;t,x,v)| \\
&\geq &|v^{\prime }-v|-CT_{0}, \\
|v^{\prime \prime }-V_{j}(s_{1};X_{i}(s;t,x,v),v^{\prime })| &\geq
&|v^{\prime \prime }-v^{\prime }|-|v^{\prime
}-V_{j}(s_{1};X_{i}(s;t,x,v),v^{\prime })| \\
&\geq &|v^{\prime \prime }-v^{\prime }|-CT_{0}.
\end{eqnarray*}
Therefore, either one of the following are valid correspondingly for some $
\sigma >0$: 
\begin{eqnarray*}
|\mathbf{k}_{w}^{i,j}(V_{i}(s),\zeta ,\xi ^{\prime })| &\leq &C_{T_{0}}e^{-
\frac{\sigma }{8}N^{2}}|\mathbf{k}_{w}^{i,j}(V_{i}(s),\zeta ,\xi ^{\prime
})e^{\frac{\sigma }{8}\{|V_{i}(s)-v^{\prime }|^{2}+|\zeta -\zeta ^{\prime
}|^{2}\}}|,\text{ \ \ } \\
\text{\ \ \ }|\mathbf{k}_{w}^{j,k}(V_{j}(s_{1}),\zeta ^{\prime },\xi
^{\prime \prime })| &\leq &C_{T_{0}}e^{-\frac{\sigma }{8}N^{2}}|\mathbf{k}
_{w}^{j,k}(V_{j}(s_{1}),\zeta ^{\prime },\xi ^{\prime \prime })\times\\&&{ e}^{\frac{
\sigma }{8}\{|V_{j}(s_{1})-v^{\prime \prime }|^{2}+|\zeta ^{\prime \prime
}-\zeta ^{\prime }|^{2}\}}|.
\end{eqnarray*}
From Lemma \ref{lemma7}, 
\begin{multline}
\int |\mathbf{k}_{w}^{i,j}(V_{i}(s),\zeta ,\xi ^{\prime })e^{\frac{\sigma }{8
}\{|V_{i}(s)-v^{\prime }|^{2}+|\zeta -\zeta ^{\prime }|^{2}\}}|d\xi ^{\prime
}\\+\int |\mathbf{k}_{w}^{j,k}(V_{j}(s_{1}),\zeta ^{\prime },\xi ^{\prime
\prime })e^{\frac{\sigma }{8}\{|V_{j}(s_{1})-v^{\prime \prime }|^{2}+|\zeta
^{\prime \prime }-\zeta ^{\prime }|^{2}\}}|d\xi ^{\prime \prime }<+\infty .
\end{multline}
We use this bound to combine the cases of $|\xi ^{\prime }-\xi |\geq N$ or $
|\xi ^{\prime }-\xi ^{\prime \prime }|\geq N$ as:
\begin{equation*}
\int_{0}^{t}\int_{0}^{s_{1}}\left\{ \int_{|\xi |\leq N,|\xi ^{\prime }|\geq
2N,\text{ \ \ }}+\int_{|\xi ^{\prime }|\leq 2N,|\xi ^{\prime \prime }|\geq
3N}\right\}
\end{equation*}
We first integrate $\xi ^{\prime }$ for the first integral and apply (\ref
{integral}) to integrate $\mathbf{k}_{w}^{j,k}$ over $\xi ^{\prime \prime }$.
We then integrate $\xi ^{\prime \prime }$ for the second integral and apply (
\ref{integral}) to integrate $\mathbf{k}_{w}^{i,j}$ over $\xi ^{\prime }$.
We thus find an upper bound 
\begin{eqnarray}\label{inflowstep3}
&&C\int_{0}^{t}\int_{0}^{s_{1}}\sup_{\xi }\int_{|\xi |\leq N,|\xi ^{\prime
}|\geq 2N,\text{ \ \ }}|\mathbf{k}_{w}^{ij}(V_{i}(s),\zeta ,\xi ^{\prime
})|d\xi ^{\prime }\notag\\&&+\sup_{\xi ^{\prime }}\int_{|\xi ^{\prime }|\leq 2N,|\xi
^{\prime \prime }|\geq 3N}|\mathbf{k}_{w}^{ij}(V_{j}(s_{1}),\zeta ^{\prime
},\xi ^{\prime \prime })|d\xi ^{\prime \prime }  \\
&\leq &\frac{C_{\eta }}{\kappa ^{2}}e^{-\frac{\eta }{8}N^{2}}\int_{0}^{t}
\int_{0}^{s_{1}}e^{-\nu _{0}(t-s_{1})}\Vert h(s_{1})\Vert _{L^{\infty
}}ds_{1}ds  \notag \\
&\leq &C_{\eta }e^{-\frac{\eta }{8}N^{2}}e^{-\frac{\nu _{0}}{2}t}\sup_{0\leq
s\leq t}\{e^{\frac{\nu _{0}}{2}s}\Vert h(s)\Vert _{L^{\infty }}\}\notag.
\end{eqnarray}

\textbf{CASE 3.} $|\xi |\leq N$, $|\xi ^{\prime }|\leq 2N,|\xi ^{\prime
\prime }|\leq 3N$. This is the last remaining case because if $|\xi ^{\prime
}|>2N$, it is included in Case 2; while if $|\xi ^{\prime \prime }|>3N$,
either $|\xi ^{\prime }|\leq 2N$ or $|\xi ^{\prime }|\geq 2N$ are also
included in Case 2. We now can bound the second term in (\ref{double}) by 
\begin{equation*}
C\int_{0}^{t}\int_{B}\int_{0}^{s}e^{-\nu _{0}(t-s_{1})}|\mathbf{k}
_{w}^{i,j}(V_{i}(s),\zeta ,\xi ^{\prime })\mathbf{k}_{w}^{j,k}(V_{j}(s_{1}),
\zeta ^{\prime },\xi ^{\prime \prime })h_{k}(s_{1},X_{j}(s_{1}),\xi ^{\prime
\prime })|
\end{equation*}
where $B=\{|\xi ^{\prime }|\leq 2N$, $|\xi ^{\prime \prime }|\leq 3N\}$. We
notice that $\mathbf{k}_{w}^{i,j}(\xi ,\xi ^{\prime })$ has a possible
integrable singularity of the type $\frac{1}{|\xi -\xi ^{\prime }|}$. We can
choose $\mathbf{k}_{N}^{i,j}(\xi ,\xi ^{\prime })$ smooth with compact
support such that 
\begin{equation}
\sup_{|p|\leq 3N}\int_{|\xi ^{\prime }|\leq 3N}|\mathbf{k}_{N}^{i,j}(p,\xi
^{\prime })-\mathbf{k}_{w}^{i,j}(p,\xi ^{\prime })|d\xi ^{\prime }\leq \frac{
1}{N}.  \label{approximate}
\end{equation}
Split $\mathbf{k}_{w}^{ij}(V_{i}(s),\zeta ,\xi ^{\prime })\mathbf{k}
_{w}^{j,k}(V_{j}(s_{1}),\zeta ^{\prime },\xi ^{\prime \prime })$ into 
\begin{eqnarray*}
&&\{\mathbf{k}_{w}^{i,j}(V_{i}(s),\zeta ,\xi ^{\prime })-\mathbf{k}
_{N}^{i,j}(V_{i}(s),\zeta ^{\prime },\xi ^{\prime })\}\mathbf{k}
_{w}^{j,k}(V_{j}(s_{1}),\zeta ^{\prime },\xi ^{\prime \prime }) \\
&&+\{\mathbf{k}_{w}^{i,j}(V_{j}(s_{1}),\zeta ,\xi ^{^{\prime }})-\mathbf{k}
_{N}^{i,j}(V_{j}(s_{1}),\zeta ,\xi ^{\prime })\}\mathbf{k}
_{N}^{j,k}(V_{j}(s),\zeta ^{\prime },\xi ^{\prime }) \\
&&+\mathbf{k}_{N}^{i,j}(V_{i}(s),\zeta ,\xi ^{\prime })\mathbf{k}
_{N}^{j,k}(V_{j}(s_{1}),\zeta ^{\prime },\xi ^{\prime \prime }).
\end{eqnarray*}
We then integrate the first term above in $\xi ^{\prime \prime }$ and the
second term above in $\xi ^{\prime }$. By (\ref{integral}), we can use such
an approximation (\ref{approximate}) to bound the $s_{1},s$ integration by 
\begin{eqnarray}
&&\frac{Ce^{-\frac{\nu _{0}}{2}t}}{N}\sup_{0\leq s\leq t}\{e^{\frac{\nu _{0}
}{2}s}\Vert h(s)\Vert _{L^{\infty }}\}  \notag \\
&&\times \left\{ \sup_{|\xi ^{\prime }|\leq 2N}\int |\mathbf{k}
_{w}^{i,j}(V_{j}(s_{1}),\zeta ^{\prime },\xi ^{\prime \prime })|d\xi
^{\prime \prime }+\sup_{|\xi |\leq 2N}\int |\mathbf{k}_{w}^{j,k}(V_{i}(s),
\zeta ,\xi ^{\prime })|d\xi ^{\prime }\}\right\}  \label{inflowstep41} \\
&&+C\int_{0}^{t}\int_{B}\int_{s_{1}}^{s}e^{-\nu _{0}(t-s_{1})}\mathbf{k}
_{N}^{ij}(V_{i}(s),\zeta ,\xi ^{\prime })\mathbf{k}_{N}^{ij}(V_{j}(s_{1}),
\xi ^{\prime \prime })|h_{j}(s,X_{j}(s_{1}),\zeta ^{\prime },\xi ^{\prime
\prime })|.  \notag
\end{eqnarray}
The first term above is further bounded by $\frac{Ce^{-\frac{\nu _{0}}{2}t}}{
N}\sup_{0\leq s\leq t}\{e^{\frac{\nu _{0}}{2}s}\Vert h(s)\Vert _{L^{\infty
}}\}$.

\ \ \ \ \ \ \ \ Fix $\varepsilon >0$. We use now Lemma \ref{determinant} for
the last main contribution in (\ref{inflowstep41}) for which we separate two
cases $|X_{j}(s;t,x,v)|\geq L_{\varepsilon }$ and $|X_{j}(s;t,x,v)|\leq
L_{\varepsilon }$ where $L_{\varepsilon }$ is given in Lemma \ref
{determinant}.

In the case $|X_{j}(s;t,x,v)|\geq L_{\varepsilon }$, we bound it by 
\begin{eqnarray*}
&&C_{N}\int_{0}^{t}\int_{B}\int_{0}^{s}e^{-\nu
_{0}(t-s_{1})}|h_{k}(s,X_{j}(s_{1}),\xi ^{\prime \prime })|\mathbf{1}
_{|X_{j}(s_{1})|\geq L_{\varepsilon }}dsd\xi ^{\prime \prime }d\xi ^{\prime
}ds_{1} \\
&\leq &C_{N}\{\int_{0}^{t-\varepsilon }+\int_{t-\varepsilon }^{t}\}.
\end{eqnarray*}
The second integral is bounded by $C\varepsilon e^{-\frac{\nu _{0}}{2}
t}\sup_{0\leq s\leq t}\{e^{\frac{\nu _{0}}{2}s}\Vert h(s)\Vert _{L^{\infty
}}\}$. In the first integral, since $s\leq t-\varepsilon $, by Lemma \ref
{determinant}, we can make a change of variable 
\begin{equation}
y=X_{j}(s_{1})=X_{j}(s_{1};X_{i}(s;t,x,v),v^{\prime })  \label{change}
\end{equation}
because $|\frac{dy}{dv^{\prime }}|\geq \frac{\varepsilon }{2}$. We observe
since $\Vert \partial _{x}\phi \Vert _{L^{\infty }}\leq C$, that from (\ref{ode})
\begin{eqnarray*}
|v^{\prime }-V_{j}(\tau )| &\leq &\int_{\tau }^{s}\Vert \partial _{x}\phi
\Vert _{L^{\infty }}d\tau \leq T_{0}\Vert \partial _{x}\phi \Vert
_{L^{\infty }}, \\
|y-X_{i}(s)| &\leq &\int_{s_{1}}^{s}|V_{j}(\tau )|d\tau \leq
T_{0}(|v^{\prime }|+T_{0}\Vert \partial _{x}\phi \Vert _{L^{\infty }})\leq
C_{T_{0},N}
\end{eqnarray*}
for $|v^{\prime }|\leq 2N$. By first integrating over $\zeta ^{\prime }$ and using
the change of variable (\ref{change}) 
\begin{eqnarray*}
&&\int_{0}^{t-\varepsilon }\int_{B}\int_{0}^{s}e^{-\nu
_{0}(t-s_{1})}|h_{k}(s_{1},X_{j}(s_{1}),\xi ^{\prime \prime })|\mathbf{1}
_{|X_{j}(s_{1})|\geq L_{\varepsilon }}dsd\xi ^{\prime \prime }d\xi ^{\prime
}ds_{1} \\
&\leq &\frac{C}{\varepsilon }\int_{0}^{t-\varepsilon }\int_{|y-X_{i}(s)|\leq
C_{T_{0},N}}\int_{|\xi ^{\prime \prime }|\leq 3N}\int_{0}^{s}e^{-\nu
_{0}(t-s_{1})}|h_{k}(s_{1},y,\xi ^{\prime \prime })|dsd\xi ^{\prime \prime
}dyds_{1} \\
&\leq &\frac{C_{N}}{\varepsilon }\sup_{0\leq s_{1}\leq
T_{0}}\int_{|y-X_{i}(s)|\leq C_{T_{0},N}}\int_{|\xi ^{\prime \prime }|\leq
3N}|h_{k}(s_{1},y,\xi ^{\prime \prime })|ds_{1}d\xi ^{\prime \prime }dy \\
&=&\int_{|f_{k}(t)-M_{k}|\geq \kappa M_{j}}+\int_{|f_{k}(t)-M_{k}|\leq
\kappa M_{j}} \\
&\leq &C_{T_{0},N,\varepsilon }\{\mathcal{H}(g(0))+\sqrt{\mathcal{H}(g(0))}
\}.
\end{eqnarray*}
We\ have used the fact $h_{k}=\frac{w(f_{k}-M_{k})}{\sqrt{M_{k}}}$, (which
is bounded by $f_{k}-M_{k}$ for $|\xi ^{\prime \prime }|\leq 3N$), and
applied Lemma \ref{entropy}.

For $|X_{i}(s;t,x,v)|\leq L_{\varepsilon }$, for any $\eta >0$, we again
employ Lemma \ref{determinant} to find $O_{x_{l}}$ such that 
\begin{eqnarray*}
&&\int_{0}^{T_{0}}\int_{B}\int_{0}^{T_{0}}e^{-\nu
_{0}(t-s_{1})}|h_{k}(s_{1},X_{j}(s_{1}),\xi ^{\prime \prime })|\mathbf{1}
_{|X_{j}(s)|\leq L_{\varepsilon }}dsd\xi ^{\prime \prime }d\xi ^{\prime
}ds_{1} \\
&=&\int_{0}^{T_{0}}\int_{B}\int_{0}^{T_{0}}\mathbf{1}_{O_{x_{l}}^{c}}e^{-\nu
_{0}(t-s_{1})}|h_{k}(s_{1},X_{j}(s_{1}),\xi ^{\prime \prime })|\mathbf{1}
_{|X_{j}(s)|\leq L_{\varepsilon }}dsd\xi ^{\prime \prime }d\xi ^{\prime
}ds_{1} \\
&+&\int_{0}^{T_{0}}\int_{B}\int_{0}^{T_{0}}\mathbf{1}_{O_{x_{l}}^{{}}}e^{-
\nu _{0}(t-s_{1})}|h_{k}(s_{1},X_{j}(s_{1}),\xi ^{\prime \prime })|\mathbf{1}
_{|X_{j}(s)|\leq L_{\varepsilon }}dsd\xi ^{\prime \prime }d\xi ^{\prime
}ds_{1}.
\end{eqnarray*}
Since $|[0,T_{0}]\times \lbrack -N,N]\cap O_{x_{l}}^{c}|<\eta $, the first
part is bounded by 
\begin{equation*}
C_{T_{0},N,\varepsilon }\eta e^{-\frac{\nu _{0}}{2}t}\sup_{0\leq s\leq
t}\{e^{\frac{\nu _{0}}{2}s}\Vert h(s)\Vert _{L^{\infty }}\}.
\end{equation*}
The second part is bounded by
\begin{equation*}
C_{T_{0},N,\varepsilon }\int_{0}^{T_{0}}\int_{0}^{T_{0}}\int_{B}\mathbf{1}
_{O_{x_{l}}}|h_{k}(s_{1},X_{j}(s_{1}),\xi ^{\prime \prime })dsds_{1}d\xi
^{\prime }d\xi ^{\prime \prime }.
\end{equation*}
Since $|\frac{\partial X_{j}(s_{1};X_{i}(s;t,x,v),v^{\prime })}{\partial
v^{\prime }}|>m_{\eta }/2$ on $O_{x_{l}}$ from Lemma \ref{determinant}, we
can make a (local) change of variable $
y=X_{j}(s_{1})=X_{j}(s_{1};X_{i}(s;t,x,v),v^{\prime })$ to get 
\begin{eqnarray*}
&&C_{T_{0},N,\varepsilon }\int_{0}^{T_{0}}\int_{0}^{T_{0}}\int_{B}\mathbf{1}
_{O_{x_{l}}^{c}}|h_{k}(s_{1},X_{j}(s_{1}),\xi ^{\prime \prime })dsds_{1}d\xi
^{\prime }d\xi ^{\prime \prime } \\
&=&C_{T_{0},N,\varepsilon
}\sum^{I_{l}}C_{T_{0},N}\int_{0}^{T_{0}}\int_{0}^{T_{0}}
\int_{|y-X_{i}(s_{1})|\leq C_{T_{0},N}}\times\\&&\int_{|\xi ^{\prime \prime }|\leq
3N}h_{k}(s_{1},y,\xi ^{\prime \prime })dsds_{1}dyd\xi ^{\prime \prime } \\
&=&\int_{|f_{k}(t)-M_{k}|\leq \kappa M_{k}}+\int_{|f_{k}(t)-M_{k}|\geq
\kappa M_{k}} \\
&\leq &C_{T_{0},N,\varepsilon ,\eta }\{\mathcal{H}(g(0))+\sqrt{\mathcal{H}
(g(0))}\}.
\end{eqnarray*}

Collecting terms, we conclude 
\begin{eqnarray*}
\sup_{0\leq s\leq T_{0}}e^{\frac{\nu _{0}}{2}t}\Vert h(s)\Vert _{L^{\infty
}} &&\leq C(1+T_{0})\Vert h(0)\Vert _{L^{\infty }} \\
&&+\{\frac{C_{T_{0}}}{N}+C_{N,T_{0}}\varepsilon +C_{N,T_{0},\varepsilon
}\eta \}\sup_{0\leq s\leq T_{0}}\{e^{\frac{\nu _{0}}{2}s}\Vert h(s)\Vert
_{L^{\infty }}\} \\
&&+C\{\sup_{0\leq s\leq T_{0}}e^{\frac{\nu _{0}}{2}s}\Vert h(s)\Vert
_{L^{\infty }}\}^{2}+C_{T_{0},N,\varepsilon ,\eta }\sqrt{\mathcal{H}(g(0))}.
\end{eqnarray*}
Assume $\sup_{0\leq s\leq T_{0}}\Vert h(s)\Vert _{L^{\infty }}$ is
sufficiently small. We first choose $T_{0}$ sufficiently large so that 
\begin{equation*}
2C(1+T_{0})e^{-\frac{\nu _{0}}{2}T_{0}}\leq \frac{1}{2},
\end{equation*}
then $N$ sufficiently large, then $\varepsilon $ sufficiently small, finally 
$\eta $ small to conclude our lemma.
\end{proof}

\begin{proof}
of Theorem \ref{stabl}: Assume $\sup_{0\leq t\leq \infty }\Vert h(t)\Vert
_{L^{\infty }}$ is small. We first establish (\ref{hstability}). Choose any $
n=0,1,2,3,\dots$ and apply Lemma \ref{bound} repeatedly to get 
\begin{eqnarray*}
\Vert h(nT_{0})\Vert _{L^{\infty }} &\leq &\frac{1}{2}\Vert
h(\{n-1\}T_{0})\Vert _{L^{\infty }}+C_{T_{0}}\sqrt{\mathcal{H}(g(0))} \\
&\leq &\frac{1}{4}\Vert h(\{n-2\}T_{0})\Vert _{L^{\infty }}+\frac{1}{2}
C_{T_{0}}\sqrt{\mathcal{H}(g(0))}+C_{T_{0}}\sqrt{\mathcal{H}(g(0))} \\
&\leq &\dots \\
&\leq &\frac{1}{2^{n}}\Vert h_{0}\Vert _{L^{\infty }}+C_{T_{0}}\sqrt{
\mathcal{H}(g(0))}\{1+\frac{1}{2}+\frac{1}{4}+\dots\} \\
&\leq &\frac{1}{2^{n}}\Vert h_{0}\Vert _{L^{\infty }}+2C_{T_{0}}\sqrt{
\mathcal{H}(g(0))}.
\end{eqnarray*}
For any $t$, we can find $n$ such that $nT_{0}\leq t\leq \{n+1\}T_{0}$, and
from $L^{\infty }$ estimate from $[0,T_{0}]$, we conclude (\ref{hstability})
by 
\begin{equation*}
\Vert h(t)\Vert _{L^{\infty }}\leq C_{T_{0}}\Vert h(nT_{0})\Vert \leq
C\{\Vert h_{0}\Vert _{L^{\infty }}+\sqrt{\mathcal{H}(g(0))}\}.
\end{equation*}

To prove (\ref{h1bound}), we take $x$ and $v$ derivatives to get 
\begin{eqnarray}
&&\{\partial _{t}+v\partial _{x}+F(M_{i+1}+\sqrt{M_{i+1}}g_{i+1})\partial
_{v}+\nu (\xi )\}\partial _{x}g_{i}-K^{i,j}\partial _{x}g_{j}  \notag \\
&=&-\partial _{x}F(M_{i+1}+\sqrt{M_{i+1}}g_{i+1})\partial _{v}g_{i}+\beta
\partial _{x}F(\sqrt{M_{i+1}}g_{i+1})v\partial _{x}\sqrt{M_{i}}+  \notag \\
&&+\partial _{x}\{F(\sqrt{M_{j}}g_{j})vg_{i}\}+\partial _{x}\{\Gamma
(g_{i},g_{i})+\Gamma (g_{i},g_{j})\};  \label{xderi} \\
&&\{\partial _{t}+v\partial _{x}+F(M_{j}+\sqrt{M_{j}}g_{j})\partial _{v}+\nu
(\xi )\}\partial _{v}g_{i}-\partial _{v}\{K^{i,j}g_{j}\}+\{\partial _{v}\nu
(\xi )\}g_{i}  \notag \\
&=&-\partial _{x}g_{i}+\beta \partial _{x}F(\sqrt{M_{j}}g_{j})v\sqrt{M_{i}}
+\beta \partial _{x}F(\sqrt{M_{j}}g_{j})v\partial _{x}\sqrt{M_{i}} \notag \\
&&+F(\sqrt{
M_{j}}g_{j})\partial _{v}\{vg_{i}\} +\partial _{v}\{\Gamma (g_{i},g_{i})+\Gamma (g_{i},g_{j})\}  \label{vderi}
\end{eqnarray}
where $K^{i,j}$ has similar property as $K_{i}$ in \cite{G1} (see Lemma 2.2
in \cite{G1}, p. 1109. In particular, $\Vert \partial
_{v}\{K^{i,j}g_{j}\}\partial _{v}g_{i}\Vert _{L^{1}}\leq \frac{1}{2}\Vert
\partial _{v}g\Vert _{\nu }^{2}+C\Vert g\Vert _{L^{2}}^{2}$ so that a
positive dissipation for $\partial _{v}g_{i}$ occur for small $\Vert h\Vert
_{L^{\infty }}$ in (\ref{vderi}). Notice that $L=\nu -K\geq 0$. We take
inner product with $\partial _{x}g_{i}$ and $\partial _{v}g_{i}$
respectively, following the procedures in \cite{G1} to get:
\begin{eqnarray}
&&\frac{d}{dt}\frac{1}{2}\Vert \partial _{x}g\Vert _{L^{2}}^{2} \leq
C\{\Vert \partial _{x}F(M_{i+1})\Vert _{L^{\infty }}+\Vert h\Vert
_{L^{\infty }}\}\Vert \nabla _{x,v}g\Vert _{L^{2}}^{2}+C\Vert g\Vert
_{L^{2}}^{2}.  \label{h1estimate} \notag\\&&
\frac{d}{dt}\frac{1}{2}\Vert \partial _{v}g\Vert _{L^{2}}^{2}+\frac{1}{4}
\Vert \partial _{v}g\Vert _{\nu }^{2} \leq C\Vert \partial _{x}g\Vert
_{L^{2}}^{2}+C\Vert g\Vert _{L^{2}}^{2}.  
\end{eqnarray}
Hence (\ref{h1bound}) follows from the Gronwall Lemma since $
\sup_{0\leq t\leq \infty }\Vert h(t)\Vert _{L^{\infty }}$ is bounded by (\ref
{hstability}). With such an estimate, we obtain the uniqueness by taking $
L^{2}$ estimate for the difference for (\ref{maineqpert}) because the most
difficult term $F(M_{i+1}+\sqrt{M_{i+1}}g_{i+1})\partial _{v}g_{i}$ can be
handled.
\end{proof}

\section{Linear Instability: Growing Mode.}

In this section we study the linearization of the equation (\ref{maineq})
around the homogeneous equilibrium $M_{\text{\rm hom}}=(\mu_{\beta},\mu_{\beta})$. In the sequel we omit the index $\beta$ for sake of shortness: $\mu =\mu _{\beta }=(\frac{\beta }{2\pi }
)^{3/2}e^{-\beta \frac{|\xi |^{2}}{2}}$. When $M$ is replaced by $M_{\text{
hom}}=(\mu ,\mu )$ in (\ref{maineqpert}), we get the following linearized
Vlasov-Boltzmann system: 
\begin{equation}
\partial _{t}{g}+\mathcal{L}{g}=0,  \label{linearvb}
\end{equation}
where ${\ g}=(g_{1},g_{2})$, 
\begin{equation*}
(\mathcal{L}{g})_{i}=v\partial _{x}g_{i}-\beta F(\sqrt{\mu }g_{i+1})v\sqrt{
\mu }-L_{i}{g}
\end{equation*}
and 
\begin{equation*}
L_{i}g=\frac{1}{\sqrt{\mu }}\Big(Q(\sqrt{\mu }g_{i},2\mu )+Q(\mu ,\sqrt{\mu }
(g_{1}+g_{2}))\Big).
\end{equation*}
We seek an exponential growing mode for such a system when $\beta >1$. To
this end, we consider a family of systems 
\begin{equation}
\partial _{t}{g}+\mathcal{L}^{\alpha }{g}=0,  \label{family}
\end{equation}
\begin{equation*}
(\mathcal{L}^{\alpha }{\ g})_{i}=v\partial _{x}g_{i}-\beta F(\sqrt{\mu }
g_{i+1})v\sqrt{\mu }-\alpha L_{i}g
\end{equation*}
and show that there is a growing mode for all $\alpha >0$.We shall first
seek a growing mode of the form $g_{1}=g_{2}=g$, so that the system (\ref
{family}) reduces to the single equation 
\begin{equation}
\{\partial _{t}+v\partial _{x}\}g-\beta F(\sqrt{\mu }g)v\sqrt{\mu }=\alpha
Lg.  \label{eigen}
\end{equation}
with $Lg=\frac{2}{\sqrt{\mu }}\Big(Q(\sqrt{\mu }g,\mu )+Q(\mu ,\sqrt{\mu }g
\Big)$. We seek a growing mode periodic in $x$, so we assume periodic
dependence on space and exponential in time: $g(t,x,\xi )=e^{\lambda t}e^{
\mathbf{i}kx}q(\xi )$. From the definition of $F$ (see (\ref{vlasovforce})),
(\ref{eigen}) becomes: 
\begin{equation}
\{\lambda +\mathbf{i}vk\}q-\beta k\mathbf{i}\hat{U}(k)\left\{ \int q\sqrt{
\mu }d\xi \right\} v\sqrt{\mu }=\alpha Lq.  \label{kmode}
\end{equation}
Equivalently, $q(\xi )$ is an eigenfunction for the operator $\mathcal{T}
^{\alpha }$ with eigenvalue $\lambda :$ 
\begin{equation}
(\mathcal{T}^{\alpha }q)(\xi )=\mathbf{i}vkq(\xi )-\beta k\mathbf{i}\hat{U}
(k)\left\{ \int q(\xi )\sqrt{\mu }d\xi \right\} v\sqrt{\mu }-\alpha Lq(\xi ).
\label{talpha}
\end{equation}

\begin{lemma}
\label{lambdalarge}Let $\beta >1$. There exists sufficiently small $\alpha
>0 $ such that there is an eigenfunction ${q(}\xi )$ to $\mathcal{T}^{\alpha
}$ with $\Re \lambda >0$.
\end{lemma}

\begin{proof}
We first study the eigenvalue problem for the \textit{unperturbed} operator $
\mathcal{T}^{0}$ for $\alpha =0:$ 
\begin{equation}
(\lambda +\mathbf{i}vk)q-\beta k\mathbf{i}\hat{U}(k)\left\{ \int_{\mathbb{R}
^{3}}q\sqrt{\mu }d\xi \right\} v\sqrt{\mu }=0.  \label{alpha=0}
\end{equation}
This is similar to the Penrose dispersion relation for the unperturbed
Vlasov-Poisson system (\cite{Pen}) for a collisionless plasma. From (\ref
{alpha=0}), we obtain: \begin{equation*}
q=\displaystyle{\frac{\beta k\hat{U}(k)\mathbf{i}\left\{ \displaystyle{\int_{
\mathbb{R}^{3}}q\sqrt{\mu }d\xi }\right\} v\sqrt{\mu }}{(\lambda +\mathbf{i}
vk)}}.
\end{equation*}
Normalizing $\displaystyle{\ \int_{\mathbb{R}^{3}}q\sqrt{\mu }d\xi =1}$, we
deduce that 
\begin{equation*}
\beta \mathbf{i}\int_{\mathbb{R}^{3}}\frac{\hat{U}(k)kv\mu (\xi )}{(\lambda +
\mathbf{i}vk)}d\xi =1.
\end{equation*}
Multiply and divide by $(\lambda -\mathbf{i}vk)$, take the imaginary part
and then integrate on $\xi $. By consistency we must have: 
\begin{equation*}
\beta \int_{\mathbb{R}^{3}}\frac{v^{2}\hat{U}(k)k^{2}\mu (\xi )}{\lambda
^{2}+k^{2}v^{2}}d\xi =1.
\end{equation*}
Define $F(\lambda ,k)\equiv \displaystyle{\beta \int_{\mathbb{R}^{3}}\frac{
v^{2}\hat{U}(k)k^{2}\mu (\xi )}{\lambda ^{2}+k^{2}v^{2}}d\xi }$. Clearly $
F(0,k)=\beta \hat{U}(k)$. Since $\hat{U}(0)=1$ and $\beta >1$, there is $
k_{0}$ sufficiently small so that 
\begin{equation}
F(0,k_{0})=\beta \hat{U}(k_{0})>1.  \label{k0}
\end{equation}
Moreover, $\lim_{\lambda \rightarrow \infty }F(\lambda ,k_{0})=0$ for any $
k_{0}\neq 0$. Hence there exists a real number (eigenvalue) $\lambda >0$
such that $F(\lambda ,k_{0})=1$ and $q={\frac{\beta k\hat{U}(k)\mathbf{i}v
\sqrt{\mu }}{(\lambda +\mathbf{i}vk)}}$ is the eigenfunction.

We now fix $k=k_{0}$ and return to (\ref{kmode}).

It can be proved (see for example \cite{G3}) that $\Vert Lq\Vert
_{L^{2}}\leq C\Vert \nu q\Vert _{L^{2}}$. Moreover, for hard spheres, there
are constants $C_{1}$ and $C_{2}>0$ such that 
\begin{eqnarray}
\Vert \mathcal{T}^{0}q\Vert _{L^{2}} &=&\Vert \mathbf{i}k_{0}vq-\beta k_{0}
\mathbf{i}\hat{U}(k_{0})\{{\int_{\mathbb{R}^{3}}q\sqrt{\mu }d\xi \}v}\sqrt{
\mu }\Vert _{L^{2}}  \notag \\
&\geq &C_{1}\Vert \nu q\Vert _{L^{2}}-C_{2}\Vert q\Vert _{L^{2}}.
\label{tlower}
\end{eqnarray}
This implies $\Vert Lq\Vert _{L^{2}}\leq C\{\Vert \mathcal{T}^{0}q\Vert
_{L^{2}}+\Vert q\Vert _{L^{2}}\}$ so the perturbation $L$ is $\mathcal{T}
^{0} $-bounded. Since 
\begin{equation*}
\mathcal{T}^{\alpha }=\mathcal{T}^{0}-\alpha L,
\end{equation*}
we deduce from Kato's book (P206, \cite{K}) that, for $\alpha $ small, there
is an eigenvalue $\lambda $ and an eigenfunction $q(v,\zeta )$ for (\ref
{kmode}) with positive real part $\lambda $ for (\ref{kmode}).
\end{proof}

\begin{theorem}
\label{lininst} Let $\beta >1$. 
Then there is a $\frac{2\pi }{k_{0}}$-periodic eigenvector $(\tilde{g}
,\tilde{g})$ for $-\mathcal{L}$ such that $\Re \lambda
>0$ and $\tilde{g}(x,v,\zeta )=\tilde{g}(-x,-v,\zeta )$.
\end{theorem}

\begin{proof}
We fix $k_{0}$ as in Lemma \ref{lambdalarge}. The idea is to define, for the
family of equations (\ref{family}) 
\begin{eqnarray*}
\alpha _{0}=\sup_{\alpha }\Big\{\alpha &:&\text{there is an eigenvalue with
positive real part for }\mathcal{T}^{\alpha } \\
&\ &\text{ with a $\frac{2\pi }{k_{0}}$~--~periodic eigenvector}\Big\}.
\end{eqnarray*}
By Lemma \ref{lambdalarge}, for $\alpha $ sufficiently small this set is not
empty. We want to show that $\alpha _{0}=+\infty $.

We prove it by contradiction. Suppose $\alpha _{0}<+\infty $. We claim that,
if there is such a finite $\alpha _{0}>0$, then there is an eigenvalue $
\lambda _{0}$ with an eigenfunction $q_{0}$ with $\Vert q_{0}\Vert _{\nu }=
\sqrt{\int \nu |q_{0}|^{2}d\xi }=1$, such that $\Re \lambda _{0}=0$ and 
\begin{equation}
(\lambda _{0}+\mathbf{i}vk_{0})q_{0}-\beta k_{0}\mathbf{i}\hat{U}
(k_{0})\left\{ \int_{\mathbb{R}^{3}}q_{0}\sqrt{\mu }d\xi \right\} v\sqrt{\mu 
}=\alpha _{0}Lq_{0}.  \label{0eigen}
\end{equation}
\textit{Proof of the claim:} In fact, by (\ref{tlower}), choose a family of
eigenfunctions $q_{\alpha }\in L^{2}$ such that $\Vert q_{\alpha }\Vert
_{\nu }=1$, $\Re \lambda _{\alpha }>0$, as $\alpha \rightarrow \alpha
_{0}$ and 
\begin{equation}
(\lambda _{\alpha }+\mathbf{i}vk_{0})q_{\alpha }-\beta k_{0}\mathbf{i}\hat{U}
(k_{0})\left\{ \int_{\mathbb{R}^{3}}q_{\alpha }\sqrt{\mu }d\xi \right\} v
\sqrt{\mu }=\alpha Lq_{\alpha }.  \label{aeigen}
\end{equation}
Notice that both $(Lq_{\alpha },\bar{q}_{\alpha })$\ and $(\mathbf{i}$ $
vk_{0}q_{\alpha },\bar{q}_{\alpha })$ are bounded by $C\Vert q_{\alpha
}\Vert _{\nu }^{2}$. As $\alpha \rightarrow \alpha _{0}$, taking $L^{2}$
inner product with $\bar{q}_{a}$ for (\ref{aeigen}), we deduce that $
|\lambda _{\alpha }|$ is bounded for $\alpha \rightarrow \alpha _{0}$. Hence 
$\lim_{\alpha \rightarrow \alpha _{0}}\lambda _{\alpha }=\lambda _{0}$ (up
to subsequences) with $\Re \lambda _{0}\geq 0$. We now prove that $
\lambda _{0}$ is an eigenvalue so that $\Re \lambda _{0}=0$ by the
definition of $\alpha _{0}$ and the claim is proven. Clearly, we may assume
that $\lim q_{\alpha }=q_{0}$ weakly in $L^{2}$ and (\ref{0eigen}) are valid
as $\alpha \rightarrow \alpha _{0}$. We only need to show that $\lim
q_{\alpha }=q_{0}$ strongly so that $\Vert q_{0}\Vert _{\nu }=1$, and $q_{0}$
is an eigenfunction. Denote $\mathbf{P}h=\{1,v,|v|^{2}\}\sqrt{\mu }$.
Clearly $\mathbf{P}q_{\alpha }\rightarrow \mathbf{P}q_{0}$ strongly in $
L^{2}$. It thus is left to show that $(\mathbf{I}-\mathbf{P})q_{\alpha
}\rightarrow (\mathbf{I}-\mathbf{P})q_{0}$ strongly in $L^{2}$. We subtract (
\ref{aeigen}) from (\ref{0eigen}) to get 
\begin{eqnarray*}
&&(\lambda _{\alpha }-\lambda _{0})q_{\alpha }+(\alpha -\alpha
_{0})Lq_{\alpha }+(\lambda _{0}+\mathbf{i}k_{0}v)(q_{\alpha }-q_{0}) \\
&&-\beta k_{0}\mathbf{i}\hat{U}(k_{0})\left\{ \int_{\mathbb{R}
^{3}}\{q_{\alpha }-q_{0}\}\sqrt{\mu }d\xi \right\} v\sqrt{\mu }+\alpha
_{0}L(q_{\alpha }-q_{0}) \\
&=&0.
\end{eqnarray*}
We take the $L^{2}$ inner product with $\bar{q}_{\alpha }-\bar{q}_{0}$ and
then take the real part. Since $\Re \lambda _{0}\Vert q_{\alpha
}-q_{0}\Vert ^{2}\geq 0$, and $\mathbf{i}k_{0}\int v|q_{\alpha
}-q_{0}|^{2}d\xi $ is purely imaginary, we obtain: 
\begin{eqnarray*}
\alpha _{0}\langle L(g^{\alpha }-g),(\bar{g}^{\alpha }-\bar{g})\rangle &\leq
&(|\lambda _{\alpha }-\lambda _{0}|+|\alpha -\alpha _{0}|)\Vert q_{\alpha
}\Vert _{\nu }\cdot \Vert q_{\alpha }-q_{0}\Vert _{\nu } \\
&&+C|\int_{\mathbb{R}^{3}}\{q_{\alpha }-q_{0}\}\sqrt{\mu }d\xi |\cdot \Vert
(g^{\alpha }-g)\Vert .
\end{eqnarray*}
Therefore, $(\mathbf{I}-\mathbf{P})\{g^{\alpha }-g\}\rightarrow 0$ in $
L_{\nu }^{2}$ and $\Vert g\Vert _{\nu }=1$ and our claim follows.

Hence, $\lambda _{0}$ is purely imaginary. Actually we show that it is $0$.
To do this we take inner product with $\bar{q}_{0}$ in (\ref{0eigen}) to get 
\begin{equation*}
\lambda _{0}\Vert q_{0}\Vert _{2}^{2}-\beta k_{0}\mathbf{i}\hat{U}
(k_{0})\left\{ \int_{\mathbb{R}^{3}}q_{0}\sqrt{\mu }d\xi \right\} \left\{
\int_{\mathbb{R}^{3}}\bar{q}_{0}v\sqrt{\mu }d\xi \right\} +\alpha
_{0}\langle Lg,\bar{g}\rangle =0.
\end{equation*}
But by integrating $\sqrt{\mu }\times $(\ref{0eigen}) over $\xi $, we obtain
the continuity equation 
\begin{equation*}
\lambda _{0}\int_{\mathbb{R}^{3}}q_{0}v\sqrt{\mu }d\xi +k_{0}\mathbf{i}\int_{
\mathbb{R}^{3}}q_{0}v\sqrt{\mu }d\xi =0.
\end{equation*}
Therefore, $k_{0}\mathbf{i}\int_{\mathbb{R}^{3}}\bar{q}_{0}v\sqrt{\mu }d\xi =
\bar{\lambda}_{0}\int_{\mathbb{R}^{3}}\bar{q}_{0}v\sqrt{\mu }d\xi $ and 
\begin{equation}
\lambda _{0}\Vert g\Vert _{2}^{2}-\beta \bar{\lambda}_{0}\hat{U}
(k_{0})\left\vert \int_{\mathbb{R}^{3}}q_{0}\sqrt{\mu }d\xi \right\vert
^{2}+\alpha _{0}\langle Lg,\bar{g}\rangle =0.  \label{real}
\end{equation}
Since $\hat{U}(k_{0})$ is real and $\lambda _{0}$ is purely imaginary,
taking the real part of (\ref{real}) we conclude that $\alpha _{0}\langle
Lq_{0},\bar{q}_{0}\rangle =0$. Therefore, $q_{0}$ is a linear combination of
the collision invariants $\sqrt{\mu },\xi \sqrt{\mu }$ and $|\xi ^{2}|^{2}
\sqrt{\mu }$~and $\alpha _{0}Lq_{0}$ vanishes in (\ref{0eigen}). Now (\ref
{0eigen}) reduces to a pure Vlasov equation, and we deduce that 
\begin{equation*}
q_{0}(\xi )=\frac{\beta \mathbf{i}k_{0}\hat{U}(k_{0})\{\int q_{0}(\xi )\sqrt{
\mu }d\xi \}v\sqrt{\mu }}{\lambda _{0}+\mathbf{i}vjk_{0}}.
\end{equation*}
Since $\int q_{0}(\xi )\sqrt{\mu }d\xi \neq 0$, this is compatible with the
condition that $q_{0}$ is a combination of collision invariants if and only
if $\lambda _{0}=0$. Thus $q_{0}(\xi )=\beta \hat{U}(k_{0})\sqrt{\mu }$ and
hence 
\begin{equation}
\beta \hat{U}(k_{0})=1  \label{equal1}
\end{equation}
which is a contradiction to (\ref{k0}).

Hence we denote $(q(v,\zeta )e^{\mathbf{i}kx},\lambda )$ a pair of
eigenfucntion and eigenvalue for the operator $\mathcal{T}^{1}$. Note that 
$\hat{U}(-k_{0})=\hat{U}(k_{0})$, and $\{Lq\}(-v,\zeta )=L\{q(-v,\zeta )\}$
from the rotation invariance of the collision operator. Letting $
v\rightarrow -v$ in $\mathcal{T}$, we get 
\begin{equation*}
\{\lambda -\mathbf{i}vk_{0}\}q(-v,\zeta )+\beta k_{0}\mathbf{i}\hat{U}
(-k_{0})\left\{ \int q(-v,\zeta )\sqrt{\mu }d\xi \right\} v\sqrt{\mu }
=L\{q(-v,\zeta )\}.
\end{equation*}

This implies that $q(-v,\zeta )e^{-\mathbf{i}kx}$ would satisfy the same (
\ref{eigen}) with the same $\lambda $. We therefore get 
\begin{eqnarray*}
&&\tilde{g}(x,\xi ) =q(v,\zeta )e^{\mathbf{i}k_{0}x}+q(-v,\zeta )e^{-\mathbf{
i}k_{0}x} \\
&=&\{\Re q(v,\zeta )+\Re q(-v,\zeta )\}\cos k_{0}x-\{\Im
q(v,\zeta )-\Im q(-v,\zeta )\}\sin k_{0}x \\
&+&\mathbf{i}\{\Re q(v,\zeta )-\Re q(-v,\zeta )\}\sin
k_{0}x+\{q(v,\zeta )+\Im q(-v,\zeta )\}\cos k_{0}x.
\end{eqnarray*}
is also an (non-zero) eigenfunction since $\int q(v,\zeta )d\xi =1$.
Clearly, $\tilde{{g}}(x,v,\zeta )=\tilde{g}(-x,-v,\zeta )$. Hence, for the
corresponding period $P=\frac{2\pi }{k_{0}}$, such that $\tilde{{g}}=(\tilde{
g}(x,\xi ),\tilde{g}(x,\xi ))$ is an eigenvector for $-\mathcal{L}$. Our
lemma thus follows.
\end{proof}

\section{Nonlinear Instability.}

In order to establish the non linear instability, we need several lemmas on
the properties of the fastest linear growing mode. First of all we need to
establish the smoothness and long time behavior for the growing mode. Recall the
definition of the operator $\mathcal{L}$ (\ref{linearvb}). We define 
\begin{equation*}
\mathcal{M}=\{{g}=[g_{1},g_{2}]\in L^{2}\,|\,g_{1}(x,v,\zeta
)=g_{2}(-x,-v,\zeta )\}
\end{equation*}
and $\|\,\cdot\,\|_{L^2{\mathcal{M}}}$ will denote the $L^2$ norm on this set. We have the following lemmas:

\begin{lemma}
\label{vidav}Let $\beta >1$. Then for $k_{0}$ sufficiently small, for all $
\delta >0$, the spectrum of $-\mathcal{L}$ in $\{\Re \lambda >\delta \}$
consists of a finite number of eigenvalues of finite multiplicity. If $
\lambda _{1}$ denotes an eigenvalue with maximal real part, and $\Lambda
>\max \{0,\Re \lambda _{1}\}$, then there exists $C_{\Lambda }>0$ such
that, for any $g_0\in \mathcal{M}$, 
\begin{equation*}
\|e^{-t\mathcal{L}}g_{0}\|_{L^{2}(\mathcal{M})}\leq C_{\Lambda }e^{\Lambda
t}\|g_{0}\|_{L^{2}(\mathcal{M})}.
\end{equation*}
\end{lemma}

\begin{proof}
This follows easily from the Vidav's Lemma \cite{V}. Notice that we can
split 
\begin{eqnarray*}
\mathcal{L}{g} &=&\mathcal{\{}v\partial _{x}{g}+L{g}\}+\{\beta F(\sqrt{\mu }{
g})v\sqrt{\mu }\} \\
&\equiv &\mathcal{A}{g}+\mathcal{K}{g}.
\end{eqnarray*}
where $\mathcal{K}$ is a compact operator from $L^{2}$ to $L^{2}$, while $
\|e^{-t\mathcal{A}}\|_{L^{2}(\mathcal{M})\rightarrow L^{2}(\mathcal{M})}\leq
1$.
\end{proof}

\begin{lemma}
\label{regularity}Let ${R}=(R_{1},R_{2})\in L^{2}(\mathcal{M})$ with $\|{R}
\|_{L^{2}(\mathcal{M})}=1$ be an eigenvector of $-\mathcal{L}$ with $\Re\lambda >0$. Then there exists a constant $C$ depending only on $\lambda $
such that 
\begin{eqnarray}
\|\nabla _{x,v}{R}\|_{L^{2}(\mathcal{M})} &\leq &C.  \label{hibound} \\
\sup_{x,v}w(v)|{R}(x,v)| &\leq &C,\text{ }  \label{rbound}
\end{eqnarray}
where $w$ is a polynomial weight as in previous section.
\end{lemma}

\begin{proof}
We begin with ${R}\in L^{2}$. We first claim 
\begin{equation}
{R}=-\int_{0}^{\infty }e^{-\lambda t}e^{-t\mathcal{A}}\mathcal{K}{R}dt.
\label{rintegral}
\end{equation}
Notice that the corresponding growing mode ${g}(t)=e^{\lambda t}{R}$
satisfies 
\begin{equation*}
\partial _{t}{g}+\mathcal{A}{g}\mathcal{=-K}{g}
\end{equation*}
so that $e^{\lambda t}{R}=e^{-(t-s)\mathcal{A}}e^{\lambda s}{R}
-\int_{s}^{t}e^{-(t-\tau )\mathcal{A}}\mathcal{K}{R}d\tau $. Letting $
s\rightarrow -\infty $, since $\|e^{-t\mathcal{A}}\|_{L^{2}\rightarrow
L^{2}}\leq 1$ for any $t>0$ and $\Re \lambda >0$, we get 
\begin{equation*}
e^{\lambda t}{R}=-\int_{-\infty }^{t}e^{-(t-\tau )\mathcal{A}}\mathcal{K}{R}
\text{ }e^{\lambda \tau }d\tau =-\int_{0}^{\infty }e^{-\tau \mathcal{A}}
\mathcal{K}{R}\text{ }e^{\lambda (t-\tau )}d\tau .
\end{equation*}
Dividing by $e^{\lambda t}$ we prove our claim because $\Re \lambda >0$
and the integral converges in $L^{2}$.

From the property of linear Boltzmann equation, clearly \\ $\Vert \partial _{x}\{e^{-t
\mathcal{A}}g\}\Vert _{L^{2}(\mathcal{M})} \leq \Vert \partial _{x}g\Vert _{L^{2}(\mathcal{M})}$.
Taking $v$ derivative of $(\partial _{t}+v\partial _{x}+L)g=0$ yields:
\begin{equation*}
\{\partial _{t}+v\partial _{x}\}\{\partial _{v}g\}+\partial
_{v}\{-Lg\}=-\partial _{x}g.
\end{equation*}
From  \cite{G1}, $\int \partial _{v}\{-Lg\}\partial _{v}g\geq \int 
\frac{\nu }{2}|\partial _{v}g|^{2}-C_{\nu }\Vert g\Vert _{L^{2}}^{2}$. We
thus obtain by taking $L^{2}$ inner product with $\partial _{v}g$, 
\begin{equation}
||\partial _{v}g||_{L^{2}}=\Vert \partial _{v}\{e^{-t\mathcal{A}
}g_{0}\}\Vert _{L^{2}}\leq C(t+1)\{\Vert g_{0}\Vert _{L^{2}}+\Vert \nabla
_{x,v}g_{0}\Vert _{L^{2}}\}.  \label{boltzmannh1}
\end{equation}
Since $\mathcal{K}{R}\in C^{\infty }$ and $\Vert \partial _{x}\{\mathcal{K}{
R\}}\Vert _{L^{2}}+\Vert \partial _{v}\{\mathcal{K}{R\}}\Vert _{L^{2}}\leq
C\Vert {R}\Vert _{L^{2}}$, we can take $\partial _{x}$ and $\partial _{v}$
derivatives in (\ref{rintegral}) to get 
\begin{eqnarray*}
\Vert \partial _{x}{R}\Vert _{L^{2}}+\Vert \partial _{v}{R}\Vert _{L^{2}}
&\leq &\int_{0}^{\infty }e^{-\Re \lambda t}\{\Vert \partial _{x}\{e^{-t
\mathcal{A}}\mathcal{K}{R\}}\Vert _{L^{2}}+\Vert \partial _{v}\{e^{-t
\mathcal{A}}\mathcal{K}{R\}}\Vert _{L^{2}}\}dt \\
&\leq &C\int_{0}^{\infty }e^{-\Re \lambda t}(t+1)\Vert \partial _{x}\{
\mathcal{K}{R\}}\Vert _{L^{2}}+\Vert \partial _{v}\{\mathcal{K}{R\}}\Vert
_{L^{2}}dt \\
&\leq &C\Vert {R}\Vert _{L^{2}}\int_{0}^{\infty }e^{-\Re \lambda
t}(t+1)dt\leq C\Vert {\ R}\Vert _{L^{2}}.
\end{eqnarray*}

We therefore deduce (\ref{hibound}). To show (\ref{rbound}), we denote ${S}
=wR$. We then have 
\begin{equation*}
\lambda {S}=\{v\partial _{x}{S}+wL(\frac{{S}}{w})\}+\{\beta F(\frac{\sqrt{
\mu }{S}}{w})vw\sqrt{\mu }\}\equiv \mathcal{A}_{w}{S}+\mathcal{K}_{w}{S}.
\end{equation*}
Applying the same proof in Section 3 for the stability for the pure linear
Boltzmann operator, we can establish: 
\begin{equation*}
\Vert e^{-t\mathcal{A}_{w}}{g}_{0}\Vert _{L^{\infty }}\leq C\{\Vert {g}
_{0}\Vert _{L^{2}}+\Vert {g}_{0}\Vert _{L^{\infty }}\}.
\end{equation*}
We can similarly obtain ${S}=-\int_{0}^{\infty }e^{-\lambda t}e^{-t\mathcal{A
}_{w}}\mathcal{K}_{w}{S}dt$, so that from $\Re \lambda >0$, 
\begin{eqnarray*}
\Vert {S}\Vert _{L^{\infty }} &\leq &\int_{0}^{\infty }e^{-\Re \lambda
t}\Vert e^{-t\mathcal{A}_{w}}\mathcal{K}_{w}{S}\Vert _{L^{\infty }}dt \\
&\leq &C\int_{0}^{\infty }e^{-\Re \lambda t}\{\Vert \mathcal{K}_{w}{S}
\Vert _{L^{\infty }}+\Vert \mathcal{K}_{w}{S}\Vert _{L^{2}}\}dt \\
&\leq &C.
\end{eqnarray*}
\end{proof}

\begin{lemma}
Let ${R}$ be an eigenvector of $-\mathcal{L}$ with its eigenvalue $\lambda$
with $\Re \lambda >0$. If $\lambda $ is not real, then there is a
constant $\zeta >0$ such that for all $t>0$, 
\begin{equation*}
\|e^{-\mathcal{L}t}\Im  {R}\|_{L^{2}}\geq \zeta e^{\Re \lambda
t}\| \Im  {R}\|_{2}>0.
\end{equation*}
\end{lemma}

\begin{proof}
We prove by contradiction. Notice that, since \\$\Im F(\sqrt{\mu}\, {R}
)=F(\sqrt{\mu}\Im  {R})$, one can immediately check that 
\begin{equation*}
e^{-\mathcal{L}t}\Im  {R}=\Im \{e^{-\mathcal{L}t} {R}\}=e^{\Re \lambda t}(\sin [\Im \lambda t]\Re  {R}+\cos [\Im 
\lambda t] \Im  {R}).
\end{equation*}
If the Lemma were false, by passing through a convergent subsequence of \\ $\sin [\Im \lambda t_{n}]$ and $\cos [\Im \lambda t_{n}]$, with $
n\rightarrow \infty $ we would have $a\Im  {R}+b\Re  {R}=0$, with $
a^{2}+b^{2}=1$. Therefore either $\Im R$ or $\Re R$ would be a
real eigenvector and $\lambda $ would be real, a contradiction.
\end{proof}

\begin{lemma}
\label{cutoff} Let ${R}$ be as in the preceeding lemma. There exists $\delta
_{0}>0$, such that for $0<\delta <\delta _{0}$, there exists a (compactly
supported) approximate eigenfunction ${R}_{\delta }$ such that 
\begin{eqnarray*}
\delta |{R}_{\delta }(x,v)|\sqrt{\mu } &\leq &\mu , \\
\Vert {R}-{R}_{\delta }\Vert _{L^{2}} &\leq &\sqrt{\delta }, \\
\Vert \partial _{x}{R}_{\delta }\Vert _{L^{2}}+\Vert \partial _{v}{R}
_{\delta }\Vert _{L^{2}} &\leq &C\{\Vert \partial _{x}{R}\Vert
_{L^{2}}+\Vert \partial _{v}{R}\Vert _{L^{2}}\}.
\end{eqnarray*}
\end{lemma}

\begin{proof}
In fact, we choose $\chi (v)$ be a smooth cutoff function $\chi (v)=1$ for $
|v|\leq N$ and $\chi (v)\equiv 0$ for $|v|\geq N+1$. By Lemma \ref
{regularity}, we have 
\begin{equation}
|\chi (v){R}(x,v)|\sqrt{\mu }\leq |{R}(x,v)|\sqrt{\mu }=|w{S}|\sqrt{\mu }
\leq \{Cw\mu ^{-1/2}\}\mu .  \label{boundR}
\end{equation}
Define $N$ by the equation $\delta =\frac{\mu ^{1/2}(N+1)}{Cw(N+1)}$ and
define 
\begin{equation*}
{R}_{\delta }=\chi (v){R}(x,v).
\end{equation*}
Clearly the third estimate in the Lemma is valid. From (\ref{boundR}) and
the definition of $N$ and $\delta $, the first inequality in the Lemma is also
valid. Since $w$ is a polynomial, we have $\sqrt{\delta }=\sqrt{\frac{\mu
^{1/2}(N+1)}{Cw(N+1)}}\geq \sqrt{\mu (N)}$ when $N$ is large. We then
conclude the lemma by 
\begin{equation*}
\|{R}-{R}_{\delta }\|_{L^{2}}=\|{R}\mathbf{1}_{|v|\geq N}\|_{L^{2}}\leq
C\int_{|v|\geq N}\frac{\mu ^{\frac{1}{2}}(v)}{w(v)}dv=C\mu ^{\frac{1}{2}
}(N)\leq \sqrt{\delta }.
\end{equation*}
\end{proof}

We now establish the crucial bootstrap Lemma which shows that $L^{2}$ growth
leads to the same growth rate for $L^{\infty }$.

\begin{lemma}
\label{boot}Let ${g}=(g_{1},g_{2})$ be a solution to the nonlinear problem
around $M_{\text{\rm hom}}$:
\begin{eqnarray}
\Big(\partial _{t}+v\partial _{x}\Big)g_{i} &+&\beta F(\sqrt{\mu }g_{i+1})
\sqrt{\mu }v-L_{i}{g}  \label{maineqpert1} \\
&=&-F(\sqrt{\mu }g_{i+1})\partial _{v}g_{i}+F(\sqrt{\mu }g_{i+1})g_{i}v+
\Gamma (g_{i},g_{i})+\Gamma (g_{i},g_{i+1}).  \notag
\end{eqnarray}
Assume that $\Re \lambda >0$ and 
\begin{equation*}
\Vert {g}(t)\Vert _{L^{2}}\leq Ce^{\Re \lambda t}\Vert {g}(0)\Vert
_{L^{2}}
\end{equation*}
for $t\in \lbrack 0,T]$. There exists $\varepsilon _{0}>0$ such that if $
\sup_{0\leq t\leq T}\{\Vert w{g}(t)\Vert _{L^{\infty }}+\Vert {g}(t)\Vert
_{L^{2}}\}\leq \varepsilon _{0}$, then there is a constant $C$ such that 
\begin{multline}
\Vert \partial _{x}{g}(t)\Vert _{L^{2}}+\Vert \partial _{v}{g}(t)\Vert
_{L^{2}}+\Vert {wg}(t)\Vert _{L^{\infty }}\\ \leq Ce^{\Re \lambda
t}\{\Vert \partial _{x}{g}(0)\Vert _{L^{2}}+\Vert \partial _{v}{g}(0)\Vert
_{L^{2}}+\Vert {h}(0)\Vert _{L^{\infty }}\}.  \label{bootstraph1}
\end{multline}
\end{lemma}

\begin{proof}
We take $x$ and $v$ derivatives for (\ref{maineqpert1}). Since $F(M_{i})=0$,
from (\ref{h1estimate}), we have $(h=wg)$ 
\begin{eqnarray}
\frac{d}{dt}\Vert \partial _{x}g\Vert _{L^{2}}^{2} &\leq &C\Vert h\Vert
_{L^{\infty }}\{\Vert \partial _{x}g\Vert _{L^{2}}^{2}+\Vert \partial
_{v}g\Vert _{L^{2}}^{2}\}+C\Vert g\Vert _{L^{2}}^{2},  \label{xvderi} \\
\frac{d}{dt}\Vert \partial _{v}g\Vert _{L^{2}}^{2} &\leq &C\Vert \partial
_{x}g\Vert _{L^{2}}^{2}+C\Vert g\Vert _{L^{2}}^{2}.  \notag
\end{eqnarray}
Applying the Gronwall's inequality to (\ref{xvderi}), by $||h||_{L^{\infty
}}\leq \varepsilon _{0}<\Re \lambda $, we obtain 
\begin{multline}
\Vert \partial _{x}g(t)\Vert _{L^{2}}^{2} \leq \Vert \partial
_{x}g(0)\Vert _{L^{2}}^{2}+C\varepsilon _{0}\int_{0}^{t}e^{C\varepsilon
_{0}(t-s)}\Vert \partial _{v}g(s)\Vert _{L^{2}}^{2}ds\\+Ce^{\Re \lambda
t}\Vert g(0)\Vert _{L^{2}}^{2},  \label{xgrowth}\end{multline}
\begin{multline}
\frac{d}{dt}\Vert \partial _{v}g(t)\Vert _{L^{2}}^{2} \leq C\varepsilon
_{0}\int_{0}^{t}e^{C\varepsilon _{0}(t-s)}\Vert \partial _{v}g(s)\Vert
_{L^{2}}^{2}ds+Ce^{\Re \lambda t}\{\Vert g(0)\Vert _{L^{2}}^{2}\\
+\Vert
\partial _{x}g(0)\Vert _{L^{2}}^{2}\}. 
\end{multline}
Letting $Y(t)=\int_{0}^{t}e^{C\varepsilon _{0}(t-s)}\Vert \partial
_{v}g(s)\Vert _{L^{2}}^{2}ds$ so that $Y^{\prime }=\Vert \partial
_{v}g(s)\Vert _{L^{2}}^{2}$. Multiplying with $Y^{\prime }\geq 0$, we have 
\begin{equation*}
\frac{d}{dt}\{Y^{\prime }\}^{2}\leq C\varepsilon _{0}\{Y^{\prime
}\}^{2}+C_{\varepsilon _{0}}e^{\Re \lambda t}\{\Vert g(0)\Vert
_{L^{2}}^{2}+\Vert \partial _{x}g(0)\Vert _{L^{2}}^{2}\}.
\end{equation*}
Gronwall's inequality and integration by parts implies 
\begin{eqnarray*}
\{Y^{\prime }\}^{2}(t) &\leq &\Vert \partial _{v}g(0)\Vert
_{L^{2}}^{4}+C\{\Vert g(0)\Vert _{L^{2}}^{2}+\Vert \partial _{x}g(0)\Vert
_{L^{2}}^{2}\}^{2}\int_{0}^{t}e^{C\varepsilon _{0}(t-s)}e^{2\Re \lambda
s}ds \\
&\leq &C\{\Vert \partial _{v}g(0)\Vert _{L^{2}}^{2}+\Vert g(0)\Vert
_{L^{2}}^{2}+\Vert \partial _{x}g(0)\Vert _{L^{2}}^{2}\}^{2}e^{2\Re 
\lambda t}.
\end{eqnarray*}

Taking the square root and plugging into (\ref{xgrowth}), we thus conclude
the desired estimate for $\Vert \partial _{v}g(s)\Vert _{L^{2}}^{2}+\Vert
\partial _{x}g(t)\Vert _{L^{2}}^{2}$ in (\ref{bootstraph1}).

We now consider the $L^{\infty }$ estimate. By repeating the same (but
simpler) argument in Lemma \ref{bound} for the full nonlinear system (with $
F(M_{i})=0$) , we can show that if $\{\Vert h(t)\Vert _{L^{\infty }}+\Vert
g(t)\Vert _{L^{2}}\}$ is sufficiently small, 
\begin{equation*}
\Vert h(t)\Vert _{L^{\infty }}\leq C\{\Vert h(0)\Vert _{L^{\infty
}}+\sup_{0\leq s\leq t}\Vert g(s)\Vert _{L^{2}}\}\leq Ce^{\Re \lambda
t}\{\Vert h(0)\Vert _{L^{\infty }}+\Vert g(0)\Vert _{L^{2}}\}.
\end{equation*}
\end{proof}

\begin{proof}[Proof of Theorem \ref{ninst}] We choose $R$ to be the eigenfunction whose
eigenvalue has the largest positive real part. If $\lambda $ is not real,
then $\Vert \Im R\Vert _{L^{2}}=r>0$. We choose the approximate
eigenfunction $\Im R^{\delta }$ to the imaginary part of $R$ by Lemma \ref
{cutoff}. In case $\lambda $ is real, we simply do not take the imaginary
parts.

We choose a family of solutions $f^{\delta }(0,x,\xi )=\mu +\delta \Im 
R^{\delta }\sqrt{\mu }\geq 0$ or $g^{\delta }(0,x,\xi )=\delta \Im 
R^{\delta }$. Note that the positivity follows from the first statement in
Lemma \ref{cutoff}, for $\delta $ sufficiently small. Clearly, from Lemma 
\ref{cutoff}, 
\begin{equation*}
\Vert g^{\delta }(0)-\delta \Im R\Vert _{L^{2}}=\delta \Vert R-\Im 
R^{\delta }\Vert _{L^{2}}\leq \delta ^{1+\frac{1}{2}}\leq \frac{\delta r}{2},
\end{equation*}
for $\delta $ sufficiently small. Hence, from Lemma \ref{cutoff}, 
\begin{equation*}
\Vert g^{\delta }(0)\Vert _{H^{1}}+\Vert h^{\delta }(0)\Vert _{L^{\infty
}}=\delta \Vert \Im R^{\delta }\Vert _{H^{1}}+\delta \Vert \Im 
R^{\delta }\Vert _{L^{\infty }}\leq C\delta r.
\end{equation*}
\ 

Now from the nonlinear Vlasov-Boltzmann system (\ref{maineqpert1}), we have 
\begin{eqnarray}
&&g^{\delta }(t)=\delta e^{-\mathcal{L}t}\Im R^{\delta }\\&&+\int_{0}^{t}e^{-
\mathcal{L(}t-\tau )}
\begin{pmatrix}
-F(\sqrt{\mu }g_{2})\partial _{v}g_{1}+F(\sqrt{\mu }g_{2})vg_{1}+\Gamma
(g_{1},g_{1})+\Gamma (g_{1},g_{2})  \label{gduhamel}\\ 
-F(\sqrt{\mu }g_{1})\partial _{v}g_{2}+F(\sqrt{\mu }g_{1})vg_{2}+\Gamma
(g_{2},g_{2})+\Gamma (g_{2},g_{1})
\end{pmatrix}
d\tau . \notag
\end{eqnarray}
We choose $\Lambda $ such that 
\begin{equation}
\Re \lambda <\Lambda <(1+\frac{1}{2})\Re \lambda .
\label{defLambda}
\end{equation}
Let 
\begin{equation*}
T^{\delta \delta }=\frac{1}{\Lambda -\Re \lambda }|\ln \frac{\zeta r}{
2C_{\Lambda }\sqrt{\delta }}|,
\end{equation*}
By (\ref{defLambda}) and $T^{\delta }=\frac{1}{\Re \lambda }\ln \frac{
\theta }{\delta }$. Since $2(\Lambda -\Re \lambda )<\Re \lambda $,
for small $\delta $, we have $T^{\delta }\leq T^{\delta \delta }$. Clearly, $
e^{(\Lambda -\Re \lambda )t}\leq e^{(\Lambda -\Re \lambda
)T^{\delta \delta }}=\frac{\zeta r}{2C_{\Lambda }\sqrt{\delta }}$, and for $
0\leq t\leq T^{\delta \delta }$: 
\begin{equation}
C_{\Lambda }\sqrt{\delta }e^{\Lambda t}\leq \frac{\zeta }{2}e^{\Re 
\lambda t}r.  \label{tdelta2}
\end{equation}
Let 
\begin{eqnarray}
T^{\ast } &=&\sup_{s}\{s:\|\nabla _{x}g^{\delta }(t)\|_{L^{2}}+\|\nabla
_{v}g^{\delta }(t)\|_{L^{2}}+\|h^{\delta }(t)\|_{L^{\infty }}\leq
\varepsilon _{0}\}  \label{tstar} \\
T^{\ast \ast } &=&\sup_{s}\{s:\|g^{\delta }(t)-\delta e^{-\mathcal{L}
t}R^{\delta }\|_{L^{2}}\leq \frac{\zeta }{4}\delta e^{\Re \lambda t}r,
\text{ for all }0\leq t\leq s\}.  \label{tstar2}
\end{eqnarray}
For $0\leq t\leq \min \{T^{\delta },T^{\ast \ast }\}$, we have from (\ref
{tdelta2}), 
\begin{eqnarray}
\|g^{\delta }(t)\|_{L^{2}} &\leq &\delta \|e^{-\mathcal{L}t}\Im 
R^{\delta }\|_{L^{2}}+\frac{\zeta }{4}\delta e^{\Re \lambda t}r
\label{l2bound} \\
&=&\delta \|e^{-\mathcal{L}t}\Im R\|_{L^{2}}+\delta \|e^{-\mathcal{L}
t}\{R-R^{\delta }\}\|_{L^{2}}+\frac{\zeta }{4}\delta e^{\Re \lambda t}r
\notag \\
&\leq &\delta e^{\Re \lambda t}+C_{\Lambda }\delta ^{\{1+\frac{1}{2}
\}}e^{\Lambda t}+\frac{\zeta }{4}\delta e^{\Re \lambda t}r  \notag \\
&\leq &(1+\frac{3\zeta }{4})e^{\Re \lambda t}\|g^{\delta }(0)\|_{L^{2}}.
\notag
\end{eqnarray}
We now claim that $T^{\delta }\leq \min \{T^{\ast },T^{\ast \ast }\}$.

In fact, if $T^{\ast }<\min \{T^{\delta },T^{\ast \ast }\}$, then by (\ref
{tdelta}), (\ref{tstar}), (\ref{l2bound}), and the  Lemma \ref{boot} (bootstrap lemma), we obtain for $0\leq t\leq T^{\ast }$: 
\begin{equation*}
\Vert \nabla _{x,v}g^{\delta }(t)\Vert _{L^{2}}+\Vert h^{\delta }(t)\Vert
_{L^{\infty }}\leq Ce^{\Re \lambda t}\{\Vert \nabla _{x,v}g^{\delta
}(0)\Vert _{L^{2}}+\Vert h^{\delta }(0)\Vert _{L^{\infty }}\}.
\end{equation*}
In particular, by (\ref{tdelta}), 
\begin{eqnarray*}
\Vert \nabla _{x,v}g^{\delta }(T^{\ast })\Vert _{L^{2}}+\Vert h^{\delta
}(T^{\ast })\Vert _{L^{\infty }} &\leq &Ce^{\Re \lambda T^{\ast
}}\{\Vert \nabla _{x,v}g^{\delta }(0)\Vert _{L^{2}}+\Vert h^{\delta
}(0)\Vert _{L^{\infty }}\} \\
&<&Ce^{\Re \lambda T^{\delta }}\delta =C\theta .
\end{eqnarray*}
This is a contradiction to the definition of $T^{\ast }$ if $\theta $ is
chosen $<<\varepsilon _{0}$.

On the other hand, if $T^{\ast \ast }<\min \{T^{\delta },T^{\ast }\}$, then
by (\ref{gduhamel}), (\ref{l2bound}), and Lemma \ref{boot}, 
\begin{eqnarray*}
\Vert g^{\delta }(T^{\ast \ast })-\delta e^{-\mathcal{L}T^{\ast \ast }}\Im R^{\delta }\Vert _{L^{2}} &\leq &C\int_{0}^{T^{\ast \ast }}\{\Vert \nabla
_{x,v}g^{\delta }(t)\Vert _{L^{2}}+\Vert h^{\delta }(t)\Vert _{L^{\infty
}}\}^{2}dt \\
&\leq &Ce^{2\Re \lambda T^{\ast \ast }}\{\Vert \nabla _{x,v}g^{\delta
}(0)\Vert _{L^{2}}+\Vert h^{\delta }(0)\Vert _{L^{\infty }}\}^{2} \\
&<&C\{e^{\Re \lambda T^{\delta }}\delta \}C\{e^{\Re \lambda
T^{\ast \ast }}\delta \} \\
&\leq &C\theta e^{\Re \lambda T^{\ast \ast }}\delta .
\end{eqnarray*}
This is a contradiction to $T^{\ast \ast }$ in (\ref{tstar2}) when
\thinspace $\theta $ is small.

Now that $T^{\delta }\leq \min \{T^{\ast },T^{\ast \ast }\}$, we can
evaluate $t=T^{\delta }$ in (\ref{gduhamel}) to get 
\begin{equation*}
\|g^{\delta }(T^{\delta })-\delta e^{-\mathcal{L}T^{\delta }}\Im 
R^{\delta }\|_{L^{2}}\leq C\theta ^{2}.
\end{equation*}
But 
\begin{eqnarray*}
\|\delta e^{-\mathcal{L}T^{\delta }}\Im R^{\delta }\|_{L^{2}} &\geq
&\|\delta e^{-\mathcal{L}T^{\delta }}\Im R\|_{L^{2}}-\|\delta e^{-
\mathcal{L}T^{\delta }}\Im \{R-R_{\delta }\}\|_{L^{2}} \\
&\geq &\zeta e^{-\Re \lambda T^{\delta }}\|\Im R\|_{L^{2}}-C_{
\Lambda }\sqrt{\delta }e^{\Lambda T^{\delta }}\delta r \\
&=&\zeta e^{\Re \lambda T^{\delta }}\delta r-\frac{\zeta }{4}e^{\Re \lambda T^{\delta }}\delta r=\frac{\zeta }{4}r\theta .
\end{eqnarray*}
Therefore, \ for $\theta $ sufficiently small, 
\begin{equation*}
\|g^{\delta }(T^{\delta })\|_{L^{2}}\geq \frac{\zeta }{4}r\theta -C\theta
^{2}\geq \frac{\zeta }{8}r\theta .
\end{equation*}
\end{proof}

\medskip

\begin{acknowledgement} R. E. and R.M. thank the Brown
University for the kind and warm hospitality. The work of R.E. and R.M. has
been supported by MURST and INDAM-GNFM. Y. G. thanks both Universit\`a di
L'Aquila and Universit\`a di Roma Tor Vergata for their hospitality when this
work was initiated. Y. G. is supported in part by an NSF grant.
\end{acknowledgement}

\end{document}